\documentclass[english]{amsart}
\usepackage[T1]{fontenc}
\usepackage[latin9]{inputenc}
\usepackage{amsthm}
\usepackage{amsmath}
\usepackage{amssymb}
\usepackage{esint}

\makeatletter
\theoremstyle{plain}
\newtheorem{thm}{\protect\theoremname}
  \theoremstyle{definition}
  \newtheorem{defn}[thm]{\protect\definitionname}
  \theoremstyle{remark}
  \newtheorem{rem}[thm]{\protect\remarkname}
  \theoremstyle{plain}
  \newtheorem{lem}[thm]{\protect\lemmaname}
  \theoremstyle{plain}
  \newtheorem{prop}[thm]{\protect\propositionname}

\DeclareMathOperator{\cat}{cat}

\makeatother

\usepackage{babel}
  \providecommand{\definitionname}{Definition}
  \providecommand{\lemmaname}{Lemma}
  \providecommand{\propositionname}{Proposition}
  \providecommand{\remarkname}{Remark}
\providecommand{\theoremname}{Theorem}

\begin{document}

\title[Solutions for KGM systems on a Riemannian manifold]{Number and profile of low energy solutions for singularly perturbed
Klein Gordon Maxwell systems on a Riemannian manifold}

\author{Marco Ghimenti}
\address[Marco Ghimenti] {Dipartimento di Matematica,
  Universit\`{a} di Pisa, via F. Buonarroti 1/c, 56127 Pisa, Italy}
\email{marco.ghimenti@dma.unipi.it.}
\author{Anna Maria Micheletti}
\address[Anna Maria Micheletti] {Dipartimento di Matematica,
  Universit\`{a} di Pisa, via F. Buonarroti 1/c, 56127 Pisa, Italy}
\email{a.micheletti@dma.unipi.it.}

\begin{abstract}
Given a 3-dimensional Riemannian manifold $(M,g)$, we investigate the existence of positive solutions of the  
Klein-Gordon-Maxwell system
$$
\left\{ \begin{array}{cc}
-\varepsilon^{2}\Delta_{g}u+au=u^{p-1}+\omega^{2}(qv-1)^{2}u & \text{in }M\\
-\Delta_{g}v+(1+q^{2}u^{2})v=qu^{2} & \text{in }M 
\end{array}\right.
$$
and  Schr\"odinger-Maxwell system 
$$
\left\{ \begin{array}{cc}
-\varepsilon^{2}\Delta_{g}u+u+\omega uv=u^{p-1} & \text{in }M\\
-\Delta_{g}v+v=qu^{2} & \text{in }M
\end{array}\right.
$$
when $p\in(4,6). $  We prove that the number of one peak solutions depends on the topological 
properties of the manifold $M$, by means of the Lusternik Schnirelmann category.
\end{abstract}

\subjclass[2010]{35J60, 35J20, 35B40,58E30,81V10}
\date{\today}
\keywords{Riemannian manifolds, Klein-Gordon-Maxwell systems, Scrh\"odinger-Maxwell systems, 
Lusternik Schnirelmann category, one peak solutions}

\maketitle

\section{Introduction}

Let $(M,g)$ be a smooth, compact, boundaryless, $3$ dimensional
Riemannian manifold.

We consider the following singularly perturbed electrostatic Klein-Gordon-Maxwell
system 

\begin{equation}
\left\{ \begin{array}{cc}
-\varepsilon^{2}\Delta_{g}u+au=|u|^{p-2}u+\omega^{2}(qv-1)^{2}u & \text{ in }M\\
\\
-\Delta_{g}v+(1+q^{2}u^{2})v=qu^{2} & \text{ in }M\\
\\
u,v>0
\end{array}\right.\label{eq:kgms}
\end{equation}
where $\varepsilon>0$, $a>0$, $q>0$, $\omega\in(-\sqrt{a},\sqrt{a})$
and $2<p<6$, and the following Schroedinger Maxwell system.
\begin{equation}
\left\{ \begin{array}{cc}
-\varepsilon^{2}\Delta_{g}u+u+\omega uv=|u|^{p-2}u & \text{ in }M\\
\\
-\Delta_{g}v+v=qu^{2} & \text{ in }M\\
\\
u,v>0
\end{array}\right.\label{eq:sms}
\end{equation}
where $\varepsilon>0$, $q>0$, $\omega>0$.

Schroedinger Maxwell systems recently received considerable attention
from the mathematical community. In the pioneering paper \cite{BF}
Benci and Fortunato studied system (\ref{eq:sms}) when $\varepsilon=1$
and without nonlinearity. Regarding the system in a semiclassical
regime Ruiz \cite{R} and D\textquoteright{}Aprile-Wei \cite{DW1}
showed the existence of a family of radially symmetric solutions respectively
for $\Omega=\mathbb{R}^{3}$ or a ball. D\textquoteright{}Aprile-Wei
\cite{DW2} also proved the existence of clustered solutions in the
case of a bounded domain $\Omega$ in $\mathbb{R}^{3}$. 

Moreover, when $\varepsilon=1$ we have results of existence and nonexistence
of solutions for pure power nonlinearities $f(v)=|v|^{p-2}v$, $2<p<6$
or in presence of a more general nonlinearity \cite{AR,ADP,AP,BJL,DM,IV,K,PS,WZ}.

In particular, Siciliano \cite{S} proves an estimate on the number
of solution for a pure power nonlinearity when $p$ is subcritical
and close to the critical exponent. 

Klein-Gordon-Maxwell systems are widely studied in physics and in
mathematical physics (see for example \cite{CB,D,KM,MN,M}. In this
setting, there are results of existence and non existence of solutions
for subcritical nonlinear terms in a bounded domain $\Omega$ (see
\cite{AP2,BF1,C,DM2,DP,DPS1,DPS2,GM,Mu}). 

As far as we know, the first result concerning the Klein-Gordon systems
on manifolds is due to Druet-Hebey \cite{DH}. They prove uniform
bounds and the existence of a solution for the system (\ref{eq:kgms})
when $\varepsilon=1$, $a$ is positive function and the exponent
$p$ is either subcritical or critical, i.e. $p\in(2,6]$. In particular,
the existence of a solution in the critical case, i.e. $p=6$, is
obtained provided the function $a$ is suitable small with respect
to the scalar curvature of the metric $g$. 

In this paper we show that the topology of the manifold $(M,g)$ has
an effect on the number of positive solutions of the systems (\ref{eq:kgms})
and (\ref{eq:sms}) with low energy. Our results read as follows.
\begin{thm}
\label{thm:1}Let $4\le p<6$. For $\varepsilon$ small enough we
have at least $\cat(M)$  non constant positive solutions of (\ref{eq:kgms})
with low energy. These solutions have a unique maximum point $P_{\varepsilon}$
and $u_{\varepsilon}=W_{\varepsilon,P_{\varepsilon}}+\Psi_{\varepsilon}$
where $W_{\varepsilon,P_{\varepsilon}}$ is defined in (\ref{eq:defweps})
$\|\Psi_{\varepsilon}\|_{L^{\infty}}\rightarrow0$ as $\varepsilon\rightarrow0$.
\end{thm}

\begin{thm}
Let $4<p<6$. For $\varepsilon$ small enough we have at least 
$\cat(M)$ non constant positive solutions of (\ref{eq:sms}) with low energy.
These solutions have a unique maximum point $P_{\varepsilon}$ and
$u_{\varepsilon}=W_{\varepsilon,P_{\varepsilon}}+\Psi_{\varepsilon}$
where $W_{\varepsilon,P_{\varepsilon}}$ is defined in (\ref{eq:defweps})
$\|\Psi_{\varepsilon}\|_{L^{\infty}}\rightarrow0$ as $\varepsilon\rightarrow0$.
\end{thm}
In \cite{GMP} the authors show that also the geometry of the manifold
$(M,g)$ has an effect on the number of positive solutions. More precisely
the authors prove that any $C^{1}$-stable critical set of the scalar
curvature $S_{g}$ of $(M,g)$ produces a solution for $2\le p<6$.

Moreover, in \cite{GMP} it is proved that generically with respect
to the metric $g$, for $\varepsilon$ small the (\ref{eq:kgms})
and the (\ref{eq:sms}) systems have at least $P_{1}(M)$ solutions
where $P_{t}(M)$ is the Poincaré polynomial of the manifold $M$
in the variable $t$ and $P_{1}(M)$ is the polynomial $P_{t}(M)$
evalued for $t=1$.

Conluding, for any metric $g$ we have at least $\cat M$ positive
low energy solutions of KGM for $4\le p<6$ and SM systems for $4<p<6$,
and generically with respect to the metric $g$, we have at least
$P_{1}(M)\ge\cat M$ positive low energy solutions of (\ref{eq:kgms})
and (\ref{eq:sms}) for $2\le p<6$.

In the following we always assume $4\le p<6$ when dealing with KGM
systems and $4<p<6$ when dealing with SM systems.

\section{Notations and definitions}

In the following we use the following notations. 
\begin{itemize}
\item $B(x,r)$ is the ball in $\mathbb{R}^{3}$ centered in $x$ with radius
$r$.
\item $B_{g}(x,r)$is the geodesic ball in $M$ centered in $x$ with radius
$r$.
\item $d_{g}(\cdot,\cdot)$is the geodesic distance in $M$.
\item The function $U(x)$ is the unique positive spherically symmetric
function in $\mathbb{R}^{3}$ such that
\[
-\Delta U+(a-\omega^{2})U=U^{p-1}\text{ in }\mathbb{R}^{3}
\]
we remark that $U$ and its first derivative decay exponentially at
infinity.
\item Given $\varepsilon>0$ we define $U_{\varepsilon}(x)=U\left(\frac{x}{\varepsilon}\right)$.
\item Let $\chi_{r}:\mathbb{R}^{+}\rightarrow\mathbb{R}^{+}$ be a smooth
cut off function such that $\chi_{r}\equiv1$ on $[0,r/2)$, $\chi_{r}\equiv0$
on $(r,+\infty)$, $|\chi'_{r}|\le2/r$ and $|\chi''_{r}|\le2/r^{2}$,
$r$ being the injectivity radius of $M$. 
\item Fixed $\xi\in M$ and $\varepsilon>0$ we define 
\begin{equation}
W_{\varepsilon,\xi}=\left\{ \begin{array}{ccc}
U_{\varepsilon}\left(\exp_{\xi}^{-1}(x)\right)\chi_{r}\left(\left|\exp_{\xi}^{-1}(x)\right|\right) &  & x\in B_{g}(\xi,r);\\
0 &  & \text{elsewhere}.
\end{array}\right.\label{eq:defweps}
\end{equation}

\item We denote by $\text{supp }\varphi$ the support of the function $\varphi$.
\item We define 
\[
m_{\infty}=\inf_{\int_{\mathbb{R}^{3}}|\nabla v|^{2}+(a-\omega^{2})v^{2}dx=|v|_{L^{p}(\mathbb{R}^{3})}^{p}}\frac{1}{2}\int_{\mathbb{R}^{3}}|\nabla v|^{2}+(a-\omega^{2})v^{2}dx-\frac{1}{p}|v|_{L^{p}(\mathbb{R}^{3})}^{p}
\]

\end{itemize}
We also use the following notation for the different norms for $u\in H_{g}^{1}(M)$:
\begin{eqnarray*}
\|u\|_{\varepsilon}^{2}=\frac{1}{\varepsilon^{3}}\int_{M}\varepsilon^{2}|\nabla_{g}u|^{2}+(a-\omega^{2})u^{2}d\mu_{g} &  & |u|_{\varepsilon,p}^{p}=\frac{1}{\varepsilon^{3}}\int_{M}|u|^{p}d\mu_{g}\\
\|u\|_{H_{g}^{1}}^{2}=\|u\|_{g}^{2}=\int_{M}|\nabla_{g}u|^{2}+u^{2}d\mu_{g} &  & \|u\|_{L_{g}^{p}}^{p}=|u|_{p,g}^{p}=\int_{M}|u|^{p}d\mu_{g}
\end{eqnarray*}
and analogously, for a function $u\in H^{1}(\mathbb{R}^{3})$
\begin{eqnarray*}
\|u\|_{H^{1}}^{2}=\int_{\mathbb{R}^{3}}|\nabla u|^{2}+u^{2}dx &  & \|u\|_{a}^{2}=\int_{\mathbb{R}^{3}}|\nabla u|^{2}+(a-\omega^{2})u^{2}dx
\end{eqnarray*}
\[
\|u\|_{L^{p}}^{p}=\int_{\mathbb{R}^{3}}|u|^{p}dx
\]
and we denote by $H_{\varepsilon}$ the Hilbert space $H_{g}^{1}(M)$
endowed with the $\|\cdot\|_{\varepsilon}$ norm.
\begin{defn}
\label{def:cat}Let $X$ a topological space and consider a closed
subset $A\subset X$. We say that $A$ has category $k$ relative
to $X$ ($\cat_{M}A=k$) if $A$ is covered by $k$ closed sets $A_{j}$,
$j=1,\dots,k$, which are contractible in $X$, and $k$ is the minimum
integer with this property. We simply denote $\cat X=\cat_{X}X$.\end{defn}
\begin{rem}
Let $X_{1}$ and $X_{2}$ be topological spaces. If $g_{1}:X_{1}\rightarrow X_{2}$
and $g_{2}:X_{2}\rightarrow X_{1}$ are continuous operators such
that $g_{2}\circ g_{1}$ is homotopic to the identity on $X_{1}$,
then $\cat X_{1}\leq\cat X_{2}$ . 
\end{rem}
We recall the following classical result 
\begin{thm}
Let $J$ be a $C^{1,1}$ real functional on a complete $C^{1,1}$
manifold $\mathcal{N}$. If $J$ is bounded from below and satisfies
the Palais Smale condition then has at least $\cat(J^{d})$ critical
point in $J^{d}$ where $J^{d}=\{u\in\mathcal{N}\ :\ J(u)<d\}$. Moreover
if $\mathcal{N}$ is contractible and $\cat J^{d}>1$, there exists
at least one critical point $u\not\in J^{d}$
\end{thm}

\section{Key estimates}

In order to overcome the problems given by the competition between
$u$ and $v$, using an idea of Benci and Fortunato \cite{BF1}, we
introduce the map $\psi:H_{g}^{1}(M)\rightarrow H_{g}^{1}(M)$ defined
by the equation
\begin{equation}
-\Delta_{g}\psi(u)+(1+q^{2}u^{2})\psi(u)=qu^{2}\text{ in case of KGM systems }\label{eq:psikgm}
\end{equation}
\begin{equation}
-\Delta_{g}\psi(u)+\psi(u)=qu^{2}\text{ in case of SM systems }\label{eq:psiSMS}
\end{equation}
The map $\psi$ is of class $C^{2}$. Its first derivative
\[
h\rightarrow\psi'(u)[h]=V_{u}(h)
\]
 is the map defined by the equation
\begin{equation}
-\Delta_{g}V_{u}(h)+(1+q^{2}u^{2})V_{u}(h)=2qu(1-q\psi(u))h\text{ in case of KGM systems}\label{eq:VhKGM}
\end{equation}
\begin{equation}
-\Delta_{g}V_{u}(h)+V_{u}(h)=2quh\text{ in case of SM systems}\label{eq:VhSMS}
\end{equation}
its second derivative $(h,k)\rightarrow\psi''(u)[h,k]=T_{u}(h,k)$
is the map defined by the equation
\begin{multline*}
-\Delta_{g}T_{u}(h,k)+(1+q^{2}u^{2})T_{u}(h,k)=-2q^{2}u(kV_{u}(h)+hV_{u}(k))+2q(1-q\psi(u))hk\\
\text{ in case of KGM systems}
\end{multline*}
\[
-\Delta_{g}T_{u}(h,k)+T_{u}(h,k)=2qhk\text{ in case of SM systems}
\]
Moreover in case of KGM systems by the maximum principle we have that
$0<\psi(u)\le1/q$, while in case of SM systems we have $\psi(u)>0$.
\begin{rem}
\label{rem:Vh}We have that $\|V_{u}(h)\|_{H_{g}^{1}}\le c|h|_{3,g}|u|_{3,g}$.
For SM systems is straightforward. In the case of KGM systems we have,
by (\ref{eq:VhKGM})
\begin{eqnarray*}
\|V_{u}(h)\|_{H_{g}^{1}}^{2} & \le & \|V_{u}(h)\|_{H_{g}^{1}}^{2}+\int_{M}q^{2}u^{2}V_{u}^{2}(h)d\mu_{g}\le\\
 & \le & \int_{M}2qu(1-q\psi(u))hV_{u}(h)d\mu_{g}\le c\|V_{u}(h)\|_{H_{g}^{1}}|h|_{3,g}|u|_{3,g}.
\end{eqnarray*}

Furthermore, for KGM systems, it holds $0\le V_{u}(u)\le2/q$ for
any $u$. (see \cite{DH})\end{rem}
\begin{lem}
\label{lem:Heb}The map $\Theta:H_{g}^{1}(M)\rightarrow\mathbb{R}$
given by 

\[
\Theta(u)=\frac{1}{2}\int_{M}(1-q\psi(u))u^{2}d\mu_{g}
\]
is $C^{1}$ and for any $u,h\in H_{g}^{1}(M)$ 
\[
\Theta'(u)[h]=\int_{M}(1-q\psi(u))^{2}uhd\mu_{g}
\]

\end{lem}
For the proof of this result we refer to \cite{DH}
\begin{lem}
\label{lem:w-psi}Let $u_{n}\rightharpoonup u$ in $H_{g}^{1}(M)$.
Then, up to subsequence, $\psi(u_{n})\rightharpoonup\psi(u)$ in $H_{g}^{1}(M)$.\end{lem}
\begin{proof}
We set $\psi_{n}:=\psi(u_{n})$. By (\ref{eq:psikgm}), it holds 
\begin{eqnarray*}
\|\psi_{n}\|_{H_{g}^{1}}^{2} & \le & \|\psi_{n}\|_{H_{g}^{1}}^{2}+\int_{M}q^{2}u_{n}^{2}\psi_{n}^{2}d\mu_{g}=q\int_{M}u_{n}^{2}\psi_{n}d\mu_{g}\le c|u_{n}|_{4,g}^{2}\|\psi_{n}\|_{H_{g}^{1}}
\end{eqnarray*}
then $\|\psi_{n}\|_{H_{g}^{1}}\le c|u_{n}|_{4,g}^{2}$, thus $\|\psi_{n}\|_{H_{g}^{1}}$
is bounded and, up to subsequence, $\psi_{n}\rightharpoonup\bar{\psi}$
in $H_{g}^{1}(M)$. We recall that $\psi_{n}$ solves (\ref{eq:psikgm}),
thus passing to the limit we have that $\bar{\psi}$ is a solution
of the equation 
\[
-\Delta_{g}\bar{\psi}+(1+q^{2}u^{2})\bar{\psi}=qu^{2}.
\]
 By the uniqueness of the solution of (\ref{eq:psikgm}) we have $\bar{\psi}=\psi(u)$.

The claim can be proved for SM systems in a similar way.\end{proof}
\begin{rem}
\label{w-1}Let $W_{\varepsilon,\xi}$ defined in (\ref{eq:defweps}).
The following limits hold uniformly with respect to $\xi\in M$.
\begin{eqnarray*}
\|W_{\varepsilon,\xi}\|_{\varepsilon}^{2} & \rightarrow & \int_{\mathbb{R}^{3}}|\nabla U|^{2}+(a-\omega^{2})U^{2}dx\\
|W_{\varepsilon,\xi}|_{\varepsilon,t} & \rightarrow & \|U\|_{L^{t}(\mathbb{R}^{3})}\text{ for all }2\le t\le6
\end{eqnarray*}
Furthermore, by definition of the function $U$, 
\[
\int_{\mathbb{R}^{3}}|\nabla U|^{2}+(a-\omega^{2})U^{2}dx=\int_{\mathbb{R}^{3}}|U|^{p}dx
\]

\end{rem}

\section{Setting of the problem}

Hereafter we limit to consider the KGM system. The case of SM system
is straightforward, the unique difference is in the proof that the
Nehari set is a regular manifold. There we have to exclude the case
$p=4$.

We consider the following functional $I_{\varepsilon}\in C^{2}(H_{g}^{1}(M),\mathbb{R})$.
\begin{equation}
I_{\varepsilon}(u)=\frac{1}{2}\|u\|_{\varepsilon}^{2}+\frac{\omega}{2}^{2}G_{\varepsilon}(u)-\frac{1}{p}|u^{+}|_{\varepsilon,p}^{p}\label{eq:ieps}
\end{equation}
where 
\[
G_{\varepsilon}(u)=\frac{1}{\varepsilon^{3}}q\int_{M}u^{2}\psi(u)d\mu_{g}
\]
and, by Lemma \ref{lem:Heb}, 
\begin{equation}
G'_{\varepsilon}(u)[\varphi]=\frac{2}{\varepsilon^{3}}\int_{M}\left(2q\psi(u)-q^{2}\psi^{2}(u)\right)u\varphi d\mu_{g}\label{eq:gprimo}
\end{equation}

The function $I_{\varepsilon}$ is of class $C^{2}$ because the map
$\psi(u)$ is of class $C^{2}$. By (\ref{eq:gprimo}) we have that
\[
I'_{\varepsilon}(u)\varphi=\frac{1}{\varepsilon^{3}}\int_{M}\varepsilon^{2}\nabla_{g}u\nabla_{g}\varphi+au\varphi-(u^{+})^{p-1}\varphi-\omega^{2}(1-q\psi(u))^{2}u\varphi d\mu_{g}
\]
thus a critical points $u_{\varepsilon}$ of the functional $I_{\varepsilon}$
is positive and it is such that the pair $(u_{\varepsilon},\psi(u_{\varepsilon}))$
is a solution of (\ref{eq:kgms}).

\section{Nehari Manifold}

We define the following Nehari set
\[
{\mathcal N}_{\varepsilon}=\left\{ u\in H_{g}^{1}(M)\smallsetminus0\ :\ N_{\varepsilon}(u):=I'_{\varepsilon}(u)[u]=0\right\} 
\]

\begin{lem}
\label{lem:nehari}${\mathcal N}_{\varepsilon}$ is a $C^{2}$ manifold
and $\inf_{{\mathcal N}_{\varepsilon}}\|u\|_{\varepsilon}>0$.\end{lem}
\begin{proof}
If $u\in{\mathcal N}_{\varepsilon}$, we have
\begin{eqnarray*}
0=N_{\varepsilon}(u) & = & \|u\|_{\varepsilon}^{2}-|u^{+}|_{\varepsilon,p}^{p}
+\frac{q\omega^{2}}{\varepsilon^{3}}\int_{M}\left(2-q\psi(u)\right)\psi(u)u^{2}d\mu_{g}\\
 & = & \|u\|_{\varepsilon}^{2}-|u^{+}|_{\varepsilon,p}^{p}+\frac{q\omega^{2}}{2\varepsilon^{3}}\int_{M}\left(2\psi(u)
+\psi'(u)[u]\right)u^{2}d\mu_{g}.
\end{eqnarray*}
The functional $N_{\varepsilon}$ is of class $C^{2}$ because $\psi$
is of class $C^{2}$. In particular we have $N'_{\varepsilon}(u)[u]<0$
for $u\in{\mathcal N}_{\varepsilon}$ and $4\le p<6$. In fact we have,
by (\ref{eq:gprimo}),

\begin{eqnarray}
N'_{\varepsilon}(u)[u] & = & 2\|u\|_{\varepsilon}^{2}-p|u^{+}|_{\varepsilon,p}^{p}+\frac{q\omega^{2}}{\varepsilon^{3}}\int_{M}\left(2-q\psi(u)\right)\psi'(u)[u]u^{2}d\mu_{g}\nonumber \\
 &  & +\frac{2q\omega^{2}}{\varepsilon^{3}}\int_{M}\left(2-q\psi(u)\right)\psi(u)u^{2}d\mu_{g}-\frac{q^{2}\omega^{2}}{\varepsilon^{3}}\int_{M}\psi'(u)[u]\psi(u)u^{2}d\mu_{g}=\nonumber \\
 & = & (2-p)\|u\|_{\varepsilon}^{2}+\frac{q\omega^{2}}{\varepsilon^{3}}\int_{M}[4-p-2q\psi(u)]\psi(u)u^{2}d\mu_{g}\nonumber \\
 &  & +\frac{q\omega^{2}}{\varepsilon^{3}}\int_{M}\left[2-\frac{p}{2}-2q\psi(u)\right]\psi'(u)[u]u^{2}d\mu_{g}<0\label{eq:nehari}
\end{eqnarray}
so ${\mathcal N}_{\varepsilon}$ is a $C^{2}$ manifold.

We prove the second claim by contradiction. Take a sequence $\left\{ u_{n}\right\} _{n}\in{\mathcal N}_{\varepsilon}$
with $\|u_{n}\|_{\varepsilon}\rightarrow0$ while $n\rightarrow+\infty$.
Thus, using that $N_{\varepsilon}(u)=0$,
\[
\|u_{n}\|_{\varepsilon}^{2}+\frac{q\omega^{2}}{\varepsilon^{3}}\int_{M}[2-q\psi(u_{n})]u_{n}^{2}\psi(u_{n})d\mu_{g}=|u_{n}^{+}|_{p,\varepsilon}^{p}\le C\|u_{n}\|_{\varepsilon}^{p},
\]
so, because $0<\psi(u_{n})<1/q$, 
\[
1\le1+\frac{q\omega^{2}}{\varepsilon^{3}\|u_{n}\|_{\varepsilon}^{2}}\int_{M}[2-q\psi(u_{n})]u_{n}^{2}\psi(u_{n})d\mu_{g}\le C\|u_{n}\|_{\varepsilon}^{p-2}\rightarrow0
\]
and this is a contradiction.\end{proof}
\begin{rem}
\label{rem:nehari}If $u\in{\mathcal N}_{\varepsilon}$, then 
\begin{align*}
I_{\varepsilon}(u) & =\left(\frac{1}{2}-\frac{1}{p}\right)\|u\|_{\varepsilon}^{2}+\left(\frac{1}{2}-\frac{2}{p}\right)\frac{\omega^{2}q}{\varepsilon^{3}}\int_{M}u^{2}\psi(u)d\mu_{g}+\frac{\omega^{2}q^{2}}{\varepsilon^{3}p}\int_{M}u^{2}\psi^{2}(u)d\mu_{g}\\
 & =\left(\frac{1}{2}-\frac{1}{p}\right)|u^{+}|_{p,\varepsilon}^{p}+\frac{1}{2}\frac{\omega^{2}q^{2}}{\varepsilon^{3}}\int_{M}u^{2}\psi^{2}(u)d\mu_{g}-\frac{1}{2}\frac{\omega^{2}q}{\varepsilon^{3}}\int_{M}u^{2}\psi(u)d\mu_{g}
\end{align*}
\end{rem}
\begin{lem}
It holds Palais-Smale condition for the functional $I_{\varepsilon}$
on the space $H_{\varepsilon}$.\end{lem}
\begin{proof}
. Let $\left\{ u_{n}\right\} _{n}\in H_{\varepsilon}$ such that 
\begin{eqnarray*}
I_{\varepsilon}(u_{n})\rightarrow c &  & \left|I'_{\varepsilon}(u_{n})[\varphi]\right|\le\sigma_{n}\|\varphi\|_{\varepsilon}\text{ where }\sigma_{n}\rightarrow0
\end{eqnarray*}
We prove that $\|u_{n}\|_{\varepsilon}$ is bounded. By contradiction,
suppose $\|u_{n}\|_{\varepsilon}\rightarrow\infty$. Then, by PS hypothesis
\begin{multline*}
\frac{pI_{\varepsilon}(u_{n})-I'_{\varepsilon}(u_{n})[u_{n}]}{\|u_{n}\|_{\varepsilon}}=\\
\left(\frac{p}{2}-1\right)\|u_{n}\|_{\varepsilon}+\frac{q\omega^{2}}{\varepsilon^{3}}\int_{M}\left[\frac{p}{2}-2+q\psi(u_{n})\right]\frac{u_{n}^{2}\psi(u_{n})}{\|u_{n}\|_{\varepsilon}}d\mu_{g}\rightarrow0
\end{multline*}
Since $p\ge4$ and $\psi(u_{n})>0$ this leads to a contradiction.

At this point, up to subsequence $u_{n}\rightarrow u$ weakly in $H_{\varepsilon}$
and strongly in $L_{g}^{t}(M)$ for each $2\le t<6$, then by Lemma
\ref{lem:w-psi} we have, up to subsequence, $\psi(u_{n}):=\psi_{n}\rightharpoonup\bar{\psi}=\psi(u)$.

We have that 
\[
u_{n}-i_{\varepsilon}^{*}[(u_{n}^{+})^{p-1}]
-\omega^{2}qi_{\varepsilon}^{*}\left[\left(q\psi_{n}^{2}-2\psi_{n}\right)u_{n}\right]\rightarrow0
\]
where the operator $i_{\varepsilon}^{*}:L_{g}^{p'},|\cdot|_{\varepsilon,p'}\rightarrow H_{\varepsilon}$
is the adjoint operator of the immersion operator $i_{\varepsilon}:H_{\varepsilon}\rightarrow L_{g}^{p},|\cdot|_{\varepsilon,p}$.
Since $i_{\varepsilon}^{*}[(u_{n}^{+})^{p-1}]$ converges to $i_{\varepsilon}^{*}[(u^{+})^{p-1}]$
it is sufficient to show that the sequence $i_{\varepsilon}^{*}\left[\left(q\psi_{n}^{2}-2\psi_{n}\right)u_{n}\right]\rightarrow i_{\varepsilon}^{*}\left[\left(q\bar{\psi}^{2}-2\bar{\psi}\right)u\right]$
in $H_{g}^{1}$ to obtain that $u_{n}\rightarrow u$ in $H_{g}^{1}$.
We will show that $\left(q\psi_{n}^{2}-2\psi_{n}\right)u_{n}\rightarrow\left(q\bar{\psi}^{2}-2\bar{\psi}\right)u$
in $L_{g}^{p'}$. We have 
\begin{equation}
|\psi_{n}u_{n}-\bar{\psi}u|_{p',g}\le|(\psi_{n}-\bar{\psi})u|_{p',g}+|\psi_{n}(u_{n}-u)|_{p',g}.\label{eq:1}
\end{equation}
and
\begin{equation}
|\psi_{n}^{2}u_{n}-\bar{\psi}^{2}u|_{p',g}\le|(\psi_{n}^{2}-\bar{\psi}^{2})u|_{p',g}+|\psi_{n}^{2}(u_{n}-u)|_{p',g}.\label{eq:2}
\end{equation}
For the first term of (\ref{eq:1}) we have

\[
\int_{M}|\psi_{n}-\bar{\psi}|^{\frac{p}{p-1}}|u|^{\frac{p}{p-1}}\le\left(\int_{M}|\psi_{n}-\bar{\psi}|^{p}\right)^{\frac{1}{p-1}}\left(\int_{M}|u|^{\frac{p}{p-2}}\right)^{\frac{p-2}{p-1}}\rightarrow0,
\]
and for the other terms we proceed in the same way. This concludes
the proof.\end{proof}
\begin{lem}
\label{lem:freePS}If $\left\{ u_{n}\right\} _{n}\in{\mathcal N}_{\varepsilon}$
is a Palais-Smale sequence for the functional $I_{\varepsilon}$ constrained
on ${\mathcal N}_{\varepsilon}$, then $\left\{ u_{n}\right\} _{n}$ is
a is a Palais-Smale sequence for the free functional $I_{\varepsilon}$
on $H_{\varepsilon}$ \end{lem}
\begin{proof}
Let $\left\{ u_{n}\right\} _{n}\in{\mathcal N}_{\varepsilon}$ such that
\[
\begin{array}{cc}
I_{\varepsilon}(u_{n})\rightarrow c\\
\left|I'_{\varepsilon}(u_{n})[\varphi]-\lambda_{n}N'(u_{n})[\varphi]\right|\le\sigma_{n}\|\varphi\|_{\varepsilon} & \text{ with }\sigma_{n}\rightarrow0
\end{array}
\]
In particular $I'_{\varepsilon}(u_{n})\left[\frac{u_{n}}{\|u_{n}\|_{\varepsilon}}\right]-\lambda_{n}N'(u_{n})\left[\frac{u_{n}}{\|u_{n}\|_{\varepsilon}}\right]\rightarrow0$.
Then 

\[
\lambda_{n}N'(u_{n})\left[\frac{u_{n}}{\|u_{n}\|_{\varepsilon}}\right]\rightarrow0.
\]
 By (\ref{eq:nehari}), if $\inf|\lambda_{n}|\ne0$, we have that
$\|u_{n}\|_{\varepsilon}\rightarrow0$ and this contradicts Lemma
\ref{lem:nehari}. 

Thus $\lambda_{n}\rightarrow0$. Since 
\[
I_{\varepsilon}(u_{n})=\left(\frac{1}{2}-\frac{1}{p}\right)\|u_{n}\|_{\varepsilon}^{2}+\left(\frac{1}{2}-\frac{2}{p}\right)\frac{\omega^{2}q}{\varepsilon^{3}}\int_{M}u_{n}^{2}\psi_{n}d\mu_{g}+\frac{\omega^{2}q^{2}}{\varepsilon^{3}p}\int_{M}u_{n}^{2}\psi_{n}^{2}d\mu_{g}\rightarrow c
\]
we have that $\|u_{n}\|_{\varepsilon}$ is bounded. By the expression
(\ref{eq:nehari}) of $N'(u_{n})$ and by Remark \ref{rem:Vh} we
have that $|N'(u_{n})[\varphi]|\le c\|\varphi\|_{\varepsilon}$. Thus
we obtain that $\left\{ u_{n}\right\} _{n}$ is a PS sequence for
the free functional $I_{\varepsilon}$, and we get the claim.\end{proof}
\begin{lem}
\label{lem:teps}For all $u\in H_{g}^{1}(M)$ such that $|u^{+}|_{\varepsilon,p}=1$
there exists a unique positive number $t_{\varepsilon}=t_{\varepsilon}(u)$
such that $t_{\varepsilon}(u)u\in{\mathcal N}_{\varepsilon}$. Moreover
$t_{\varepsilon}(u)$ depends continuosly on $u$, provided that $u^{+}\not\equiv0$.
Finally it holds
\[
\lim_{\varepsilon\rightarrow0}t_{\varepsilon}(W_{\varepsilon,\xi})=1\text{ uniformly with respect to }\xi\in M.
\]
\end{lem}
\begin{proof}
We define, for $t>0$ 
\[
H(t)=I_{\varepsilon}(tu)=\frac{1}{2}t^{2}\|u\|_{\varepsilon}^{2}+\frac{q\omega^{2}}{2\varepsilon^{3}}t^{2}\int_{M}\psi(tu)u^{2}d\mu_{g}-\frac{t^{p}}{p}.
\]
Thus, by (\ref{eq:gprimo})
\begin{eqnarray}
H'(t) & = & t\left(\|u\|_{\varepsilon}^{2}+\frac{q\omega^{2}}{2\varepsilon^{3}}\int_{M}[2-q\psi(tu)]\psi(tu)u^{2}d\mu_{g}-t^{p-2}\right)\label{eq:Hprimo}\\
 & = & t\left(\|u\|_{\varepsilon}^{2}+\frac{q\omega^{2}}{\varepsilon^{3}}\int_{M}\psi(tu)u^{2}d\mu_{g}+\frac{q\omega^{2}}{2\varepsilon^{3}}t\int_{M}\psi'(tu)[u]u^{2}d\mu_{g}-t^{p-2}\right)\nonumber \\
H''(t) & = & \|u\|_{\varepsilon}^{2}+\frac{q\omega^{2}}{2\varepsilon^{3}}\int_{M}[2-q\psi(tu)]\psi(tu)u^{2}d\mu_{g}\label{eq:Hsec}\\
 &  & +\frac{q\omega^{2}}{\varepsilon^{3}}t\int_{M}[1-q\psi(tu)]\psi'(tu)[u]u^{2}d\mu_{g}-(p-1)t^{p-2}\nonumber 
\end{eqnarray}
By (\ref{eq:Hprimo}) there exists $t_{\varepsilon}>0$ such that
$H'(t_{\varepsilon})=0$, because, for small $t$, $H'(t)>0$ and,
since $p\ge4$, it holds $H'(t)<0$ for $t$ large. For $t_{\varepsilon}$,
by (\ref{eq:Hprimo}) we have
\[
t_{\varepsilon}^{p-2}=\|u\|_{\varepsilon}^{2}+\frac{q\omega^{2}}{\varepsilon^{3}}\int_{M}\psi(t_{\varepsilon}u)u^{2}d\mu_{g}+\frac{q\omega^{2}}{2\varepsilon^{3}}t_{\varepsilon}\int_{M}\psi'(t_{\varepsilon}u)[u]u^{2}d\mu_{g}
\]
 then, by Remark \ref{rem:Vh} 
\begin{eqnarray*}
H''(t_{\varepsilon}) & = & (2-p)\|u\|_{\varepsilon}^{2}+\frac{q\omega^{2}}{\varepsilon^{3}}\int_{M}\left[2-p-\frac{q}{2}\psi(t_{\varepsilon}u)\right]\psi(t_{\varepsilon}u)u^{2}d\mu_{g}\\
 &  & +\frac{q\omega^{2}}{2\varepsilon^{3}}\int_{M}[3-p-2q\psi(t_{\varepsilon}u)]\psi'(t_{\varepsilon}u)[t_{\varepsilon}u]u^{2}d\mu_{g}<0,
\end{eqnarray*}
so $t_{\varepsilon}$ is unique. The continuity of $t_{\varepsilon}$
is standard.

We now prove the last claim. We have 
\begin{multline}
t_{\varepsilon}^{p-2}|W_{\varepsilon,\xi}|_{\varepsilon,p}^{p}=\|W_{\varepsilon,\xi}\|_{\varepsilon}^{2}+\frac{q\omega^{2}}{\varepsilon^{3}}\int_{M}\psi(t_{\varepsilon}W_{\varepsilon,\xi})W_{\varepsilon,\xi}^{2}d\mu_{g}\\
-\frac{q^{2}\omega^{2}}{2\varepsilon^{3}}\int_{M}\psi^{2}(t_{\varepsilon}W_{\varepsilon,\xi})W_{\varepsilon,\xi}^{2}d\mu_{g}\label{eq:teps1-1}
\end{multline}
where $t_{\varepsilon}=t_{\varepsilon}(W_{\varepsilon,\xi})$. It
holds
\begin{eqnarray}
 &  & \lim_{\varepsilon\rightarrow0}\frac{1}{\varepsilon^{3}t_{\varepsilon}^{2}}\int_{M}\psi(t_{\varepsilon}W_{\varepsilon,\xi})W_{\varepsilon,\xi}^{2}d\mu_{g}=0\label{eq:lim1}\\
 &  & \lim_{\varepsilon\rightarrow0}\frac{1}{\varepsilon^{3}t_{\varepsilon}^{4}}\int_{M}\psi^{2}(t_{\varepsilon}W_{\varepsilon,\xi})W_{\varepsilon,\xi}^{2}d\mu_{g}=0\label{eq:lim2}
\end{eqnarray}
In fact, set $\psi(t_{\varepsilon}W_{\varepsilon,\xi}):=\psi_{\varepsilon}$.
We have, by Remark \ref{w-1} 
\begin{eqnarray*}
\|\psi_{\varepsilon}\|_{H_{g}^{1}}^{2} & \le & \int_{M}|\nabla\psi_{\varepsilon}|^{2}+\psi_{\varepsilon}^{2}(1+q^{2}t_{\varepsilon}^{2}W_{\varepsilon,\xi}^{2})d\mu_{g}=t_{\varepsilon}^{2}q\int_{M}W_{\varepsilon,\xi}^{2}\psi_{\varepsilon}d\mu_{g}\le\\
 & \le & ct_{\varepsilon}^{2}|\psi_{\varepsilon}|_{6,g}\left(\int_{M}W_{\varepsilon,\xi}^{12/5}d\mu_{g}\right)^{5/6}\le ct_{\varepsilon}^{2}\|\psi_{\varepsilon}\|_{H_{g}^{1}}\varepsilon^{5/2}.
\end{eqnarray*}
Moreover
\[
\frac{1}{\varepsilon^{3}}\int_{M}\psi_{\varepsilon}W_{\varepsilon,\xi}^{2}d\mu_{g}\le\frac{1}{\varepsilon^{3}}\|\psi_{\varepsilon}\|_{H_{g}^{1}}\left(\int_{M}W_{\varepsilon,\xi}^{12/5}d\mu_{g}\right)^{5/6}\le ct_{\varepsilon}^{2}\frac{1}{\varepsilon^{3}}\varepsilon^{5},
\]
 and 
\[
\frac{1}{\varepsilon^{3}}\int_{M}\psi_{\varepsilon}^{2}W_{\varepsilon,\xi}^{2}d\mu_{g}\le\frac{1}{\varepsilon^{3}}\left(\int\psi_{\varepsilon}^{6}d\mu_{g}\right)^{1/3}\left(\int_{M}W_{\varepsilon,\xi}^{3}d\mu_{g}\right)^{2/3}\le\frac{1}{\varepsilon^{3}}\|\psi_{\varepsilon}\|_{H_{g}^{1}}^{2}\varepsilon^{2}\le t_{\varepsilon}^{4}\varepsilon^{4}
\]
so we proved (\ref{eq:lim1}) and (\ref{eq:lim2}). 

For any sequence $\varepsilon_{n}\rightarrow0$, by (\ref{eq:lim1})
and (\ref{eq:lim2}) and by Remark \ref{w-1} we have that $t_{\varepsilon_{n}}$
is bounded. Then, up to subsequences $t_{\varepsilon_{n}}\rightarrow\bar{t}$.
By (\ref{eq:teps1-1}) we have $\bar{t}^{p-2}|U|_{L^{p}(\mathbb{R}^{3})}^{p}=\int_{\mathbb{R}^{3}}|\nabla U|^{2}+(a-\omega^{2})U^{2}dx$
and by Remark \ref{w-1} we have $\bar{t}=1$.
\end{proof}

\section{Main ingredient of the proof}

We sketch the proof of Theorem \ref{thm:1}. First of all, since the
functional $I_{\varepsilon}\in C^{2}$ is bounded below and satisfies
PS condition on the manifold ${\mathcal N}_{\varepsilon}$, we have, by
well known results, that $I_{\varepsilon}$ has at least $\cat I_{\varepsilon}^{d}$
critical points in the sublevel
\[
I_{\varepsilon}^{d}=\left\{ u\in{\mathcal N}_{\varepsilon}\ :\ I_{\varepsilon}(u)\le d\right\} .
\]
We prove that, for $\varepsilon$ and $\delta$ small enough, it holds
\[
\cat M\le\cat\left({\mathcal N}_{\varepsilon}\cap I_{\varepsilon}^{m_{\infty}+\delta}\right)
\]
where 
\[
m_{\infty}:=\inf_{{\mathcal N}_{\infty}}\frac{1}{2}\int_{\mathbb{R}^{3}}|\nabla v|^{2}+(a-\omega^{2})v^{2}dx-\frac{1}{p}\int_{\mathbb{R}^{3}}|v|^{p}dx
\]
\[
{\mathcal N}_{\infty}=\left\{ v\in H^{1}(\mathbb{R}^{3})\smallsetminus\left\{ 0\right\} \ :\ \int_{\mathbb{R}^{3}}|\nabla v|^{2}+(a-\omega^{2})v^{2}dx=\int_{\mathbb{R}^{3}}|v|^{p}dx\right\} .
\]
To get the inequality $\cat M\le\cat\left({\mathcal N}_{\varepsilon}\cap I_{\varepsilon}^{m_{\infty}+\delta}\right)$
we build two continuous operators
\begin{eqnarray*}
\Phi_{\varepsilon} & : & M\rightarrow{\mathcal N}_{\varepsilon}\cap I_{\varepsilon}^{m_{\infty}+\delta}\\
\beta & : & {\mathcal N}_{\varepsilon}\cap I_{\varepsilon}^{m_{\infty}+\delta}\rightarrow M_{r}
\end{eqnarray*}
where $M_{r}=\left\{ x\in\mathbb{R}^{N}\ :\ d(x,M)<r\right\} $ with
$r$ small enough in order to have $\cat M=\cat M_{r}$. Also, we
will choose $r$ smaller than the injectivity radius of $M$. 

We build these operators $\Phi_{\varepsilon}$ and $\beta$ such that
$\beta\circ\Phi_{\varepsilon}:M\rightarrow M_{r}$ is homotopic to
the immersion $i:M\rightarrow M_{r}$. By the properties of Lusternik
Schinerlmann category we have
\[
\cat M\le\cat\left({\mathcal N}_{\varepsilon}\cap I_{\varepsilon}^{m_{\infty}+\delta}\right)
\]
which gives us the estimates on the number of solutions contained
in Theorem \ref{thm:1}. With respect to the profile description of
any low energy solution $u_{\varepsilon}$, first, we prove that $u_{\varepsilon}$
has a unique local maximum point $P_{\varepsilon}$ (see Lemma \ref{lem:unimax})
then we show that $u_{\varepsilon}=W_{P_{\varepsilon},\varepsilon}+\Psi_{\varepsilon}$
where $\|\Psi_{\varepsilon}\|_{L^{\infty}(M)}\rightarrow0$ for $\varepsilon\rightarrow0$
(see Lemma \ref{lem:stime}).

\section{The function $\Phi_{\varepsilon}$}

We define a map
\begin{eqnarray*}
\Phi_{\varepsilon} & : & M\rightarrow{\mathcal N}_{\varepsilon}\\
\Phi_{\varepsilon}(\xi) & = & t_{\varepsilon}(W_{\xi,\varepsilon})W_{\xi,\varepsilon}
\end{eqnarray*}

\begin{prop}
\label{prop:phieps}For all $\varepsilon>0$ the map $\Phi_{\varepsilon}$
is continuous. Moreover for any $\delta>0$ there exists $\varepsilon_{0}=\varepsilon_{0}(\delta)$
such that, if $\varepsilon<\varepsilon_{0}$ then $I_{\varepsilon}\left(\Phi_{\varepsilon}(\xi)\right)<m_{\infty}+\delta$.\end{prop}
\begin{proof}
It is easy to see that $\Phi_{\varepsilon}$ is continuous because
$t_{\varepsilon}(w)$ depends continously on $w\in H_{g}^{1}$.

Now, we have 
\[
I_{\varepsilon}\left(t_{\varepsilon}(W_{\varepsilon,\xi})W_{\varepsilon,\xi}\right)=\frac{1}{2}t_{\varepsilon}^{2}\|W_{\varepsilon,\xi}\|_{\varepsilon}^{2}-\frac{1}{p}t_{\varepsilon}^{p}|W_{\varepsilon,\xi}|_{\varepsilon,p}^{p}+\frac{1}{\varepsilon^{3}}qt_{\varepsilon}^{2}\int_{M}\psi(t_{\varepsilon}W_{\varepsilon,\xi})W_{\varepsilon,\xi}^{2}d\mu_{g}
\]
By Remark \ref{w-1} and Lemma \ref{lem:teps} and by (\ref{eq:lim1})
we have

\[
I_{\varepsilon}\left(t_{\varepsilon}(W_{\varepsilon,\xi})W_{\varepsilon,\xi}\right)\rightarrow m_{\infty}
\]
uniformly with respect to $\xi$. This concludes the proof.\end{proof}
\begin{rem}
\label{rem:limsup}We set
\[
m_{\varepsilon}=\inf_{{\mathcal N}_{\varepsilon}}I_{\varepsilon.}
\]
By Proposition \ref{prop:phieps} we have that 

\[
\limsup_{\varepsilon\rightarrow0}m_{\varepsilon}\le m_{\infty.}
\]

\end{rem}

\section{The map $\beta$}

For any $u\in{\mathcal N}_{\varepsilon}$ we can define a point $\beta(u)\in\mathbb{R}^{N}$
by 
\[
\beta(u)=\frac{\int_{M}x|u^{+}|^{p}dx}{\int_{M}|u^{+}|^{p}dx}.
\]
The function $\beta$ is well defined in ${\mathcal N}_{\varepsilon}$
because, if $u\in{\mathcal N}_{\varepsilon}$, then $u^{+}\neq0$.

We have to prove that, if $u\in{\mathcal N}_{\varepsilon}\cap I_{\varepsilon}^{m_{\infty}+\delta}$
then $\beta(u)\in M_{R}$.

Let us consider partitions of the compact manifold $M$. For a given
$\varepsilon>0$ we say that a finite partition ${\mathcal P}_{\varepsilon}=\left\{ P_{j}^{\varepsilon}\right\} _{j\in\Lambda_{\varepsilon}}$
of the manifold $M$ is a \textquotedblleft{}good\textquotedblright{}
partition if: for any $j\in\Lambda_{\varepsilon}$ the set $P_{j}^{\varepsilon}$
is closed; $P_{i}^{\varepsilon}\cap P_{j}^{\varepsilon}\subset\partial P_{i}^{\varepsilon}\cap\partial P_{j}^{\varepsilon}$
for any $i\ne j$; there exist $r_{1}(\varepsilon),r_{2}(\varepsilon)>0$
such that there are points $q_{j}^{\varepsilon}\in P_{j}^{\varepsilon}$
for which  $B_{g}(q_{j}^{\varepsilon},\varepsilon)\subset P_{j}^{\varepsilon}\subset B_{g}(q_{j}^{\varepsilon},r_{2}(\varepsilon))\subset B_{g}(q_{j}^{\varepsilon},r_{1}(\varepsilon))$,
with $r_{1}(\varepsilon)\ge r_{2}(\varepsilon)\ge C\varepsilon$ for
some positive constant $C$; lastly, there exists a finite number
$\nu(M)\in\mathbb{N}$ such that every $\xi\in M$ is contained in
at most $\nu(M)$ balls $B_{g}(q_{j}^{\varepsilon},r_{1}(\varepsilon))$,
where $\nu(M)$ does not depends on $\varepsilon$.
\begin{lem}
\label{lem:gamma}There exists a constant $\gamma>0$ such that, for
any $\delta>0$ and for any $\varepsilon<\varepsilon_{0}(\delta)$
as in Proposition \ref{prop:phieps}, given any ``good'' partition
${\mathcal P}_{\varepsilon}=\left\{ P_{j}^{\varepsilon}\right\} _{j}$
of the manifold $M$ and for any function $u\in{\mathcal N}_{\varepsilon}\cap I_{\varepsilon}^{m_{\infty}+\delta}$
there exists, for an index $\bar{j}$ a set $P_{\bar{j}}^{\varepsilon}$
such that 
\[
\frac{1}{\varepsilon^{3}}\int_{P_{\bar{j}}^{\varepsilon}}|u^{+}|^{p}dx\ge\gamma.
\]
\end{lem}
\begin{proof}
Taking in account that $N_{\varepsilon}(u)=I'(u)[u]=0$ we have
\begin{eqnarray*}
\|u\|_{\varepsilon}^{2} & = & |u^{+}|_{\varepsilon,p}^{p}-\frac{q\omega^{2}}{\varepsilon^{3}}\int_{M}\left(2-q\psi(u)\right)\psi(u)u^{2}d\mu_{g}\le|u^{+}|_{\varepsilon,p}^{p}\\
 & = & \sum_{j}\frac{1}{\varepsilon^{3}}\int_{P_{j}}|u^{+}|^{p}d\mu_{g}=\sum_{j}|u_{j}^{+}|_{\varepsilon,p}^{p}\\
 & = & \sum_{j}|u_{j}^{+}|_{\varepsilon,p}^{p-2}|u_{j}^{+}|_{\varepsilon,p}^{2}\le\max_{j}\left\{ |u_{j}^{+}|_{\varepsilon,p}^{p-2}\right\} \sum_{j}|u_{j}^{+}|_{\varepsilon,p}^{2}
\end{eqnarray*}
where $u_{j}^{+}$ is the restriction of the function $u^{+}$on the
set $P_{j}$. 

At this point, arguing as in Lemma 5.3 of \cite{BBM}, we prove that
there exists a constant $C>0$ such that 
\[
\sum_{j}|u_{j}^{+}|_{\varepsilon,p}^{2}\le C\nu(M)\|u^{+}\|_{\varepsilon}^{2},
\]
thus 
\[
\max_{j}\left\{ |u_{j}^{+}|_{\varepsilon,p}^{p-2}\right\} \ge\frac{1}{C\nu(M)}
\]
 that conludes the proof.\end{proof}
\begin{prop}
\label{prop:baricentro}For any $\eta\in(0,1)$ there exists $\delta_{0}<m_{\infty}$
such that for any $\delta\in(0,\delta_{0})$ and any $\varepsilon\in(0,\varepsilon_{0}(\delta))$
as in Proposition \ref{prop:phieps}, for any function $u\in{\mathcal N}_{\varepsilon}\cap I_{\varepsilon}^{m_{\infty}+\delta}$
we can find a point $q=q(u)\in M$ such that 
\[
\frac{1}{\varepsilon^{3}}\int_{B(q,r/2)}(u^{+})^{p}>\left(1-\eta\right)\frac{2p}{p-2}m_{\infty}.
\]
\end{prop}
\begin{proof}
We first prove the proposition for $u\in{\mathcal N}_{\varepsilon}\cap I_{\varepsilon}^{m_{\varepsilon}+2\delta}$. 

By contradiction, we assume that there exists $\eta\in(0,1)$ such
that we can find two sequences of vanishing real number $\delta_{k}$
and $\varepsilon_{k}$ and a sequence of functions $\left\{ u_{k}\right\} _{k}$
such that $u_{k}\in{\mathcal N}_{\varepsilon_{k}}$, 
\begin{multline}
m_{\varepsilon_{k}}\le I_{\varepsilon_{k}}(u_{k})=\left(\frac{1}{2}-\frac{1}{p}\right)\|u_{k}\|_{\varepsilon_{k}}^{2}+\\
+\left(\frac{1}{2}-\frac{2}{p}\right)\frac{\omega^{2}q}{\varepsilon_{k}^{3}}\int_{M}u_{k}^{2}\psi(u_{k})d\mu_{g}+\frac{\omega^{2}q^{2}}{\varepsilon_{k}^{3}p}\int_{M}u_{k}^{2}\psi^{2}(u_{k})d\mu_{g}\\
\le m_{\varepsilon_{k}}+2\delta_{k}\le m_{\infty}+3\delta_{k}\label{eq:mepsk}
\end{multline}
 for $k$ large enough (see Remark \ref{rem:limsup}), and, for any
$q\in M$, 
\begin{equation}
\frac{1}{\varepsilon_{k}^{3}}\int_{B_{g}(q,r/2)}(u_{k}^{+})^{p}d\mu_{g}\le\left(1-\eta\right)\frac{2p}{p-2}m_{\infty}.\label{eq:contradiction}
\end{equation}
By Ekeland principle and by Lemma \ref{lem:freePS} we can assume
\begin{equation}
\left|I'_{\varepsilon_{k}}(u_{k})[\varphi]\right|\le\sigma_{k}\|\varphi\|_{\varepsilon_{k}}\text{ where }\sigma_{k}\rightarrow0.\label{eq:ps}
\end{equation}

By Lemma \ref{lem:gamma} there exists a set $P_{k}^{\varepsilon_{k}}\in{\mathcal P}_{\varepsilon_{k}}$
such that 
\[
\frac{1}{\varepsilon_{k}^{3}}\int_{P_{k}^{\varepsilon_{k}}}|u_{k}^{+}|^{p}d\mu_{g}\ge\gamma.
\]
 and we choose a point $q_{k}\in P_{k}^{\varepsilon_{k}}$. We define
\[
u_{k}(x)\chi_{r}\left(\left|\exp_{q_{k}}^{-1}(x)\right|\right)=u_{k}(\exp_{q_{k}}(y))\chi_{r}(|y|)=u_{k}(\exp_{q_{k}}(\varepsilon_{k}z))\chi_{r}(\varepsilon_{k}|z|):=w_{k}(z)
\]
where $x\in B_{g}(q_{k},r)$ and $z\in B(0,r/\varepsilon_{k})\subset\mathbb{R}^{3}$. 

We have that $w_{k}\in H_{0}^{1}(B(0,r/\varepsilon_{k}))\subset H^{1}(\mathbb{R}^{3})$.
By equation (\ref{eq:mepsk}) we have 
\[
\|w_{k}\|_{H^{1}(\mathbb{R}^{3})}^{2}\le C\|u_{k}\|_{\varepsilon_{k}}^{2}\le C.
\]
So there exists a $w\in H^{1}(\mathbb{R}^{3})$ such that, up to subsequences,
$w_{k}\rightarrow w$ weakly in $H^{1}(\mathbb{R}^{3})$ and strongy
in $L_{\text{loc}}^{t}(\mathbb{R}^{3})$. 

We claim that $w\ge0$ is a weak solution of 
\begin{equation}
-\Delta w+(a-\omega^{2})w=w^{p-1}\label{eq:weq}
\end{equation}
We notice that for every $f\in C_{0}^{\infty}(\mathbb{R}^{3})$, there
exists $k$ such that $\text{supp}f\subset B(0,r/2\varepsilon_{k})$.
Thus, considered 
\[
f_{k}(x):=f\left(\frac{1}{\varepsilon_{k}}\exp_{q_{k}}^{-1}(x)\right)=f(z)\text{ where }x=\exp_{q_{k}}(\varepsilon_{k}z),
\]
we have that $\text{supp}f_{k}\subset B_{g}(q_{k},r/2)$. 

Moreover, we have $\|f_{k}\|_{\varepsilon_{k}}\le C\|f\|_{H^{1}(\mathbb{R}^{3})}$,
thus, by Ekeland principle we have 
\begin{equation}
|I'_{\varepsilon_{k}}(u_{k})[f_{k}]|\le\sigma_{k}\|f_{k}\|_{\varepsilon_{k}}\rightarrow0\text{ while }k\rightarrow\infty.\label{eq:stella}
\end{equation}
On the other hand we have 
\begin{multline}
I'_{\varepsilon}(u_{k})[f_{k}]=\frac{1}{\varepsilon_{k}^{3}}\int_{M}\varepsilon_{k}^{2}\nabla_{g}u_{k}\nabla_{g}f_{k}+au_{k}f_{k}-(u_{k}^{+})^{p-1}f_{k}-\omega^{2}(1-q\psi(u_{k}))^{2}u_{k}f_{k}d\mu_{g}\\
=\left\langle u_{k},f_{k}\right\rangle _{\varepsilon_{k}}-\frac{1}{\varepsilon_{k}^{3}}\int_{M}(u_{k}^{+})^{p-1}f_{k}d\mu_{g}+\frac{q\omega^{2}}{\varepsilon_{k}^{3}}\int_{M}\left(2-q\psi(u_{k})\right)\psi(u_{k})u_{k}f_{k}d\mu_{g}\\
=\int_{T_{k}}\left[\sum_{ij}g_{q_{k}}^{ij}(\varepsilon_{k}z)\partial_{z_{i}}w_{k}(z)\partial_{z_{j}}f(z)+(a-\omega^{2})w_{k}(z)f(z)\right]|g_{q_{k}}(\varepsilon z)|^{1/2}dz\\
-\int_{T_{k}}(w_{k}^{+}(z))^{p-1}f(z)|g_{q_{k}}(\varepsilon z)|^{1/2}dz\\
+q\omega^{2}\int_{T_{k}}\left(2-q\tilde{\psi}_{k}(z)\right)\tilde{\psi}_{k}(z)w_{k}(z)f(z)|g_{q_{k}}(\varepsilon z)|^{1/2}dz\label{eq:Iprimouk-1}
\end{multline}

Here $T_{k}=B(0,r/2\varepsilon_{k})\cap\text{supp}f$ and $\psi(u_{k})(x):=\psi_{k}(x)=\psi_{k}(\exp_{q_{k}}(\varepsilon_{k}z)):=\tilde{\psi}_{k}(z)$
where $x\in B_{g}(q_{k},r)$ and $z\in B(0,r/\varepsilon_{k})$. Since
$\text{supp}f_{k}\subset B_{g}(q_{k},r/2)$, by definition of $\chi_{r}$
and by (\ref{eq:psikgm}) we have
\begin{eqnarray*}
0 & = & \int_{M}\nabla_{g}\psi(u_{k})\nabla_{g}f_{k}+(1+q^{2}u_{k})\psi(u_{k})f_{k}-qu_{k}^{2}f_{k}d\mu_{g}\\
 & = & \frac{\varepsilon_{k}^{3}}{\varepsilon_{k}^{2}}\int_{T_{k}}\sum_{ij}g_{q_{k}}^{ij}(\varepsilon_{k}z)\partial_{z_{i}}\tilde{\psi}_{k}(z)\partial_{z_{j}}f(z)|g_{q_{k}}(\varepsilon z)|^{1/2}dz\\
 &  & +\varepsilon_{k}^{3}\int_{T_{k}}(1+q^{2}w_{k}(z))\tilde{\psi}_{k}(z)f(z)|g_{q_{k}}(\varepsilon z)|^{1/2}dz\\
 &  & -\varepsilon_{k}^{3}\int_{T_{k}}qw_{k}^{2}(z)f(z)|g_{q_{k}}(\varepsilon z)|^{1/2}dz,
\end{eqnarray*}
so
\begin{multline}
\int_{T_{k}}\sum_{ij}g_{q_{k}}^{ij}(\varepsilon_{k}z)\partial_{z_{i}}\tilde{\psi}_{k}(z)\partial_{z_{j}}f(z)|g_{q_{k}}(\varepsilon z)|^{1/2}dz=\\
=\varepsilon_{k}^{2}\int_{T_{k}}\left((1+q^{2}w_{k}(z))\tilde{\psi}_{k}(z)+qw_{k}^{2}(z)\right)f(z)|g_{q_{k}}(\varepsilon z)|^{1/2}dz\label{eq:psitildek}
\end{multline}

Arguing as in Lemma \ref{lem:w-psi} we have that 
\begin{eqnarray*}
c\int_{B(0,r/\varepsilon_{k})}|\nabla\tilde{\psi}_{k}(z)|^{2}dz & \le & \frac{1}{\varepsilon_{k}}\int_{M}|\nabla_{g}\psi_{k}|^{2}d\mu_{g}\le\frac{1}{\varepsilon_{k}}q\int_{M}u_{k}^{2}\psi_{k}\\
 & \le & \frac{1}{\varepsilon_{k}^{3}}\int u_{k}^{2}\le\|u_{k}\|_{\varepsilon_{k}}^{2}\le C
\end{eqnarray*}
where $c,C>0$ are suitable constants. Moreover, by Lemma \ref{lem:w-psi}
\begin{eqnarray*}
c_{1}\int_{B(0,r/\varepsilon_{k})}|\tilde{\psi}_{k}(z)|^{2}dz & \le & \frac{1}{\varepsilon_{k}^{3}}\int_{M}\psi_{k}^{2}d\mu_{g}\le\|\psi_{k}\|_{H_{g}^{1}}^{2}\le c_{2}\frac{1}{\varepsilon_{k}^{3}}|u_{k}|_{4,g}^{4}\\
 & \le & c_{2}|u_{k}|_{4,\varepsilon}^{4}\le C
\end{eqnarray*}
where $c_{1},c_{2},C>0$ are suitable constants. Conlcuding, we have
that $\|\tilde{\psi}_{k}\|_{H^{1}(B(0,r/\varepsilon_{k}))}$ is bounded,
and then also $\|\chi_{r/\varepsilon_{k}}(z)\tilde{\psi}_{k}(z)\|_{H^{1}(\mathbb{R}^{3})}^{2}$
is bounded. So, there exists a $\bar{\psi}\in H^{1}(\mathbb{R}^{3})$
such that $\bar{\psi}_{k}(z):=\chi_{r/\varepsilon_{k}}(z)\tilde{\psi}_{k}(z)\rightarrow\bar{\psi}$
weakly in $H^{1}(\mathbb{R}^{3})$ and strongly in $L_{\text{loc}}^{p}(\mathbb{R}^{3})$
for any $2\le p<6$. 

By (\ref{eq:psitildek}) we have
\begin{multline*}
\int_{\mathbb{R}^{3}}\sum_{ij}g_{q_{k}}^{ij}(\varepsilon_{k}z)\partial_{z_{i}}\bar{\psi}_{k}(z)\partial_{z_{j}}f(z)|g_{q_{k}}(\varepsilon z)|^{1/2}dz=\\
=\varepsilon_{k}^{2}\int_{\mathbb{R}^{3}}\left((1+q^{2}w_{k}(z))\bar{\psi}_{k}(z)+qw_{k}^{2}(z)\right)f(z)|g_{q_{k}}(\varepsilon z)|^{1/2}dz
\end{multline*}
and, using that $g_{k}^{ij}(\varepsilon z)=\delta_{ij}+O(\varepsilon_{k}^{2}|z|)$
and that $|g_{q}(\varepsilon z)|^{1/2}=1+O(\varepsilon_{k}^{2}|z|)$
we get 
\begin{multline*}
\int_{\mathbb{R}^{3}}\nabla\bar{\psi}_{k}(z)\nabla f(z)dz=\varepsilon_{k}^{2}\int_{\mathbb{R}^{3}}\left((1+q^{2}w_{k}(z))\bar{\psi}_{k}(z)+qw_{k}^{2}(z)\right)f(z)dz+O(\varepsilon_{k}^{2}).
\end{multline*}
Thus, the function $\bar{\psi}\in H^{1}(\mathbb{R}^{3})$ is a weak
solution of $-\Delta\bar{\psi}=0$, so $\bar{\psi}=0$.

At this point, arguing as above we have
\begin{multline}
\frac{1}{\varepsilon_{k}^{3}}\int_{M}\left(2-q\psi(u_{k})\right)\psi(u_{k})u_{k}f_{k}d\mu_{g}=\frac{1}{\varepsilon_{k}^{3}}\int_{B_{g}(q_{k},r/2)}\left(2-q\psi(u_{k})\right)\psi(u_{k})u_{k}f_{k}d\mu_{g}=\\
=\int_{\text{supp}f}\left(2-q\bar{\psi}_{k}\right)\bar{\psi}_{k}w_{k}f|g_{q_{k}}(\varepsilon z)|^{1/2}dz\rightarrow0\label{eq:conv0}
\end{multline}
while $k\rightarrow\infty$ because $\bar{\psi}_{k}\rightarrow0$
strongly in $L_{\text{loc}}^{p}(\mathbb{R}^{3})$ for any $2\le p<6$.
Thus, by (\ref{eq:conv0}), (\ref{eq:stella}) and (\ref{eq:Iprimouk-1})
and because $w_{k}\rightharpoonup w$ in $H^{1}$ we deduce that,
for any $f\in C_{0}^{\infty}(\mathbb{R}^{3})$, it holds 
\[
\int_{\mathbb{R}^{3}}\nabla w\nabla f+(a-\omega^{2})wf-(w^{+})^{p-1}f=0.
\]
Thus, $w$ is a weak solution of $-\Delta w+(a-\omega^{2})w=w^{p-1}$.
Moreover, by (\ref{eq:mepsk}) we have 
\begin{equation}
\|w\|_{a}^{2}\le\liminf_{k}\|w_{k}\|_{a}^{2}\le\frac{2p}{p-2}m_{\infty}\label{eq:normaa}
\end{equation}

Set 
\[
\mathcal{N}_{\infty}=\left\{ v\in H^{1}(\mathbb{R}^{3})\smallsetminus\left\{ 0\right\} \ :\ \|w\|_{a}^{2}=|w|_{p}^{p}\right\} ,
\]
we have that $w\in\mathcal{N}_{\infty}\cup\left\{ 0\right\} $. We
want to prove now that $w\not\equiv0$. In fact, by the definition
of the partition ${\mathcal P}_{\varepsilon}$ we can choose a $T>0$
such that $P_{k}^{\varepsilon}\subset B_{g}(q_{k},\varepsilon_{k}T)$.
Now, by Lemma \ref{lem:gamma} we have, for $k$ large
\begin{eqnarray}
\int_{B(0,T)}\left(w_{k}^{+}\right)^{p}dz & = & \int_{B(0,T)}\left(u_{k}^{+}(\exp_{q_{k}}(\varepsilon_{k}z))\chi_{r}(\varepsilon_{k}|z|)\right)^{p}dz\nonumber \\
 & \ge & \frac{C}{\varepsilon_{k}^{3}}\int_{B(0,\varepsilon_{k}T)}\left(u_{k}^{+}(\exp_{q_{k}}(y))\right)^{p}|g_{q_{k}}(y)|^{1/2}dy\nonumber \\
 & \ge & C\frac{1}{\varepsilon_{k}^{3}}\int_{P_{k}^{\varepsilon_{k}}}|u_{k}^{+}|^{p}d\mu_{g}\ge\gamma.\label{eq:contogamma}
\end{eqnarray}
Thus, $w\in\mathcal{N}_{\infty}$ and, we have, in light of (\ref{eq:normaa})
that $\|w\|_{a}^{2}=|w|_{p}^{p}=\frac{2p}{p-2}m_{\infty}$ and that
that $ $$w_{k}\rightarrow w$ strongly in $H^{1}(\mathbb{R}^{3})$.

Arguing as in (\ref{eq:contogamma}), and remembering that $|g_{q}(\varepsilon_{k}z)|^{1/2}=1+O(\varepsilon_{k}^{2}|z|)$,
fixed $T$, by (\ref{eq:contradiction}) we get, for $k$ large

\[
\int_{B(0,T)}\left(w_{k}^{+}\right)^{p}dz\le\left(1-\frac{\eta}{2}\right)\frac{2p}{p-2}m_{\infty}.
\]
Moreover, there exists a $T>0$ such that $\int_{B(0,T)}w^{p}dz>\left(1-\frac{\eta}{8}\right)\frac{2p}{p-2}m_{\infty}$
and, for $w_{k}\rightarrow w$ strongly in $L_{\text{loc}}^{p}(\mathbb{R}^{n})$,
$\int_{B(0,T)}\left(w_{k}^{+}\right)^{p}dz>\left(1-\frac{\eta}{4}\right)\frac{2p}{p-2}m_{\infty}$.
This gives us the contradiction. At this point we have proved the
claim for $u\in{\mathcal N}_{\varepsilon}\cap I_{\varepsilon}^{m_{\varepsilon}+2\delta}$.
By the thesis for $2u\in{\mathcal N}_{\varepsilon}\cap I_{\varepsilon}^{m_{\varepsilon}+\delta}$
we can prove the claim in the general case. Indeed, for $u_{k}$ it
holds

\begin{eqnarray*}
I_{\varepsilon_{k}}(u_{k}) & = & \left(\frac{1}{2}-\frac{1}{p}\right)|u_{k}^{+}|_{p,\varepsilon_{k}}^{p}+\frac{1}{2}\frac{\omega^{2}q^{2}}{\varepsilon_{k}^{3}}\int_{M}u_{k}^{2}\psi_{k}^{2}(u_{k})d\mu_{g}-\frac{1}{2}\frac{\omega^{2}q}{\varepsilon_{k}^{3}}\int_{M}u_{k}^{2}\psi(u_{k})d\mu_{g}\\
 & \ge & (1-\eta)m_{\infty}-\frac{1}{2}\frac{\omega^{2}q}{\varepsilon_{k}^{3}}\int_{M}u_{k}^{2}\psi(u_{k})d\mu_{g}
\end{eqnarray*}
By compactness of $M$ there exists $q_{1},\dots,q_{l}$ such that
\[
\frac{1}{\varepsilon_{k}^{3}}\int_{M}u_{k}^{2}\psi(u_{k})d\mu_{g}\le\sum_{i=1}^{l}\frac{1}{\varepsilon_{k}^{3}}\int_{B_{g}(q_{i},r)}u_{k}^{2}\psi(u_{k})d\mu_{g}
\]
For any $q_{i}$, arguing as above, we can introduce two sequences
of functions $w_{k}^{i}$ and $\bar{\psi}_{k}$ such that $w_{k}^{i}\rightarrow w^{i}$,
strongly in $H^{1}(\mathbb{R}^{3})$, $w^{i}$ solution of (\ref{eq:weq}),
and that $\bar{\psi}_{k}^{i}\rightarrow0$ strongly in $L_{\text{loc}}^{p}(\mathbb{R}^{3})$
for any $2\le p<6$. We thus have that, for any $q^{i}$ 
\[
\frac{1}{\varepsilon_{k}^{3}}\int_{B_{g}(q^{i},r)}u_{k}^{2}\psi(u_{k})d\mu_{g}\le\int_{\mathbb{R}^{3}}\left(w_{k}^{i}\right)^{2}\bar{\psi}_{k}^{i}dx\rightarrow0.
\]
At this point we have that $\limsup_{k}m_{\varepsilon_{k}}\ge m_{\infty}$,
and, in light of Remark \ref{rem:limsup}, that $\lim_{k}m_{\varepsilon_{k}}=m_{\infty}$.
Hence, when $\varepsilon,\delta$ are small enough, ${\mathcal N}_{\varepsilon}\cap I_{\varepsilon}^{m_{\infty}+\delta}\subset{\mathcal N}_{\varepsilon}\cap I_{\varepsilon}^{m_{\varepsilon}+2\delta}$
and the general claim follows.\end{proof}
\begin{prop}
There exists $\delta_{0}\in(0,m_{\infty})$ such that for any $\delta\in(0,\delta_{0})$
and any $\varepsilon\in(0,\varepsilon(\delta_{0}))$ (see Proposition
\ref{prop:phieps}), for every function $u\in{\mathcal N}_{\varepsilon}\cap I_{\varepsilon}^{m_{\infty}+\delta}$
it holds $\beta(u)\in M_{r}$. Moreover the composition 
\[
\beta\circ\Phi_{\varepsilon}:M\rightarrow M_{r}
\]
 is s homotopic to the immersion $i:M\rightarrow M_{r}$ \end{prop}
\begin{proof}
By Proposition \ref{prop:baricentro}, for any function $u\in{\mathcal N}_{\varepsilon}\cap I_{\varepsilon}^{m_{\infty}+\delta}$,
for any $\eta\in(0,1)$ and for $\varepsilon,\delta$ small enough,
we can find a point $q=q(u)\in M$ such that 
\[
\frac{1}{\varepsilon^{3}}\int_{B(q,r/2)}(u^{+})^{p}>\left(1-\eta\right)\frac{2p}{p-2}m_{\infty}.
\]
Moreover, since $u\in{\mathcal N}_{\varepsilon}\cap I_{\varepsilon}^{m_{\infty}+\delta}$
we have 
\begin{align*}
m_{\infty}+\delta & \ge I_{\varepsilon}(u)=\left(\frac{p-2}{2p}\right)|u^{+}|_{p,\varepsilon}^{p}+\frac{\omega^{2}q^{2}}{2\varepsilon^{3}}\int_{M}u^{2}\psi^{2}(u)d\mu_{g}-\frac{\omega^{2}q}{2\varepsilon^{3}}\int_{M}u^{2}\psi(u)d\mu_{g}\ge\\
 & \ge\left(\frac{p-2}{2p}\right)|u^{+}|_{p,\varepsilon}^{p}-\frac{\omega^{2}q}{2\varepsilon^{3}}\int_{M}u^{2}\psi(u)d\mu_{g}
\end{align*}

Now, arguing as in Lemma \ref{lem:w-psi} we have that $\|\psi(u)\|_{H^{1}(M)}\le\left(\int_{M}u^{12/5}\right)^{5/6}$,
then 
\begin{eqnarray*}
\frac{1}{\varepsilon^{3}}\int_{M}\psi(u)u^{2} & \le & \frac{1}{\varepsilon^{3}}\|\psi\|_{H^{1}(M)}\left(\int_{M}u^{12/5}\right)^{5/6}\le C\frac{1}{\varepsilon^{3}}\left(\int_{M}u^{12/5}\right)^{5/3}\\
 & \le & C\varepsilon^{2}|u|_{12/5,\varepsilon}^{4}\le C\varepsilon^{2}\|u\|_{\varepsilon}^{4}\le C\varepsilon^{2}
\end{eqnarray*}
because $\|u\|_{\varepsilon}$ is bounded since $u\in{\mathcal N}_{\varepsilon}\cap I_{\varepsilon}^{m_{\infty}+\delta}$.

Hence, provided we choose $\varepsilon(\delta_{0})$ small enough,
we have 
\[
\left(\frac{p-2}{2p}\right)|u^{+}|_{p,\varepsilon}^{p}\le m_{\infty}+2\delta_{0}.
\]
 So, 
\[
\frac{\frac{1}{\varepsilon^{3}}\int_{B(q,r/2)}(u^{+})^{p}}{|u^{+}|_{p,\varepsilon}^{p}}>\frac{1-\eta}{1+2\delta_{0}/m_{\infty}}
\]
Finally, 
\begin{eqnarray*}
|\beta(u)-q| & \le & \frac{\left|\frac{1}{\varepsilon^{3}}\int_{M}(x-q)(u^{+})^{p}\right|}{|u^{+}|_{p,\varepsilon}^{p}}\\
 & \le & \frac{\left|\frac{1}{\varepsilon^{3}}\int_{B(q,r/2)}(x-q)(u^{+})^{p}\right|}{|u^{+}|_{p,\varepsilon}^{p}}+\frac{\left|\frac{1}{\varepsilon^{3}}\int_{M\smallsetminus B(q,r/2)}(x-q)(u^{+})^{p}\right|}{|u^{+}|_{p,\varepsilon}^{p}}\\
 & \le & \frac{r}{2}+2\text{diam}(M)\left(1-\frac{1-\eta}{1+2\delta_{0}/m_{\infty}}\right),
\end{eqnarray*}
where $\text{diam}(M)$ is the diameter of the manifold as a subset
of $\mathbb{R}^{N}$. Thus, choosing $\eta$, $\delta_{0}$ and $\varepsilon(\delta_{0})$
small enough we proved the first claim. The second claim is standard.
\end{proof}

\section{\label{sec:Profile-description}Profile description}

Let $u_{\varepsilon}$ a low energy solution. By standard regularity
theory we can prove that $u_{\varepsilon}\in C^{2}(M)$. So there
exists at least one maximum point of $u_{\varepsilon}$ on $M$. We
can prove that, for $\varepsilon$ small, $u_{\varepsilon}$ has a
unique local maximum point $P_{\varepsilon}$ and we can describe
the profile of $u_{\varepsilon}$.
\begin{lem}
Let $(u_{\varepsilon},\psi(u_{\varepsilon}))$ be solution of (\ref{eq:kgms})
such that $I_{\varepsilon}(u_{\varepsilon})\le m_{\infty}+\delta<2m_{\infty}$.
Then, for $\varepsilon$ small, $u_{\varepsilon}$ is not constant
on $M$.\end{lem}
\begin{proof}
At first we notice that if $u_{\varepsilon}$ is constant, also $\psi(u_{\varepsilon})$
is constant. Moreover, by (\ref{eq:kgms}) the values of $u_{\varepsilon}$
and $\psi(u_{\varepsilon})$ depend only on $a,\omega,q$ and $p$.
Let $u_{\varepsilon}=u_{0}$ and $\psi(u_{\varepsilon})=\psi_{0}$.
Immediatly we have
\begin{multline*}
I_{\varepsilon}(u_{\varepsilon})=\left(\frac{1}{2}-\frac{1}{p}\right)\frac{1}{\varepsilon^{3}}\int_{M}(a-\omega^{2})u_{0}^{2}d\mu_{g}\\
+\left(\frac{1}{2}-\frac{2}{p}\right)\frac{\omega^{2}q}{\varepsilon^{3}}\int_{M}u_{0}^{2}\psi_{0}d\mu_{g}+\frac{\omega^{2}q^{2}}{\varepsilon^{3}p}\int_{M}u_{0}^{2}\psi_{0}^{2}d\mu_{g}\rightarrow+\infty
\end{multline*}
which leads us to a contradiction.\end{proof}
\begin{lem}
Let $x_{0}\in M$ be a maximum point for $u_{\varepsilon}$ solution
of (\ref{eq:kgms}). Then 
\begin{equation}
\left(u_{\varepsilon}(x_{0})\right)^{p-2}>a-\omega^{2}\label{eq:maxval}
\end{equation}
\end{lem}
\begin{proof}
Since $x_{0}$ is a maximum point, $\Delta_{g}u_{\varepsilon}(x_{0})\le0$.
Thus 
\[
0\ge\varepsilon^{2}\Delta_{g}u_{\varepsilon}(x_{0})=u_{\varepsilon}(x_{0})\left[a-\left(u_{\varepsilon}(x_{0})\right)^{p-2}-\omega^{2}\left(q\psi(u_{\varepsilon})(x_{0})-1\right)^{2}\right]
\]
and, recalling that $|q\psi(u_{\varepsilon})-1|<1$,
\[
a\le\left(u_{\varepsilon}(x_{0})\right)^{p-2}+\omega^{2}\left(q\psi(u_{\varepsilon})(x_{0})-1\right)^{2}\le\left(u_{\varepsilon}(x_{0})\right)^{p-2}+\omega^{2}
\]
that concludes the proof.\end{proof}
\begin{lem}
\label{lem:2max}Let $u_{\varepsilon}$ be a solution of (\ref{eq:kgms})
such that $I_{\varepsilon}(u_{\varepsilon})\le m_{\infty}+\delta<2m_{\infty}$.
Suppose that $u_{\varepsilon}$ for every $\varepsilon$ has two maximum
points $P_{\varepsilon}^{1},P_{\varepsilon}^{2}\in M$. Then, when
$\varepsilon$ is sufficiently small, 
\[
d_{g}(P_{\varepsilon}^{1},P_{\varepsilon}^{2})\rightarrow0\text{ as }\varepsilon\rightarrow0.
\]
\end{lem}
\begin{proof}
By contradiction, let $\left\{ \varepsilon_{j}\right\} _{j}$ a vanishing
sequence such that $P_{\varepsilon_{j}}^{1}\rightarrow P^{1}\in M$
and $P_{\varepsilon_{j}}^{2}\rightarrow P^{2}\in M$ and $P^{1}\neq P^{2}$.
We define $Q_{\varepsilon_{j}}^{i}\in\mathbb{R}^{n}$ such that
\[
P_{\varepsilon_{j}}^{i}=\exp_{P^{i}}(Q_{\varepsilon_{j}}^{i})\ \ i=1,2.
\]

we define 
\[
0<v^{1}(z):=u_{\varepsilon_{j}}\left(\exp_{P^{1}}(Q_{\varepsilon_{j}}^{1}+\varepsilon_{j}z)\right)=u_{\varepsilon_{j}}\left(\exp_{P^{1}}(y)\right)=u_{\varepsilon_{j}}(x)
\]
where $x\in B_{g}(P^{1},r)$ and $z\in\mathbb{R}^{n}$ such that $|Q_{\varepsilon_{j}}^{1}+\varepsilon_{j}z|<r$,
$r$ being the injectivity radius of $M$. We notice that, by definition,
$Q_{\varepsilon_{j}}^{1}\rightarrow0$ as $\varepsilon_{j}\rightarrow0$,
so in the following we simply assume $|z|<R/\varepsilon_{j}$ for
the sake of simplicity.

By (\ref{eq:maxval}) we have that $v^{1}(0)=u_{\varepsilon_{j}}(P_{\varepsilon_{j}}^{1})\ge a-\omega^{2}>0$,
moreover
\[
c\|v_{j}^{1}\|_{H^{1}(B(0,r/\varepsilon_{j}))}^{2}\le\|u_{\varepsilon_{j}}\|_{\varepsilon_{j}}^{2}\le\frac{2p}{p-2}I_{\varepsilon_{j}}(u_{\varepsilon_{j}})\le2m_{\infty}.
\]
We define
\[
\tilde{v}_{j}^{1}(z)=v_{j}^{1}(z)\chi_{r}(|Q_{\varepsilon_{j}}^{1}+\varepsilon_{j}z|)=u_{\varepsilon_{j}}(x)\chi_{r}(\exp_{P_{1}}^{-1}(x))\in H^{1}(\mathbb{R}^{3}).
\]
We have that 
\[
\|\tilde{v}_{j}^{1}\|_{H^{1}(\mathbb{R}^{3})}\le c\|v_{j}^{1}\|_{H^{1}(B(0,r/\varepsilon_{j}))}^{2}\le C
\]
thus there exists $\tilde{v}^{1}\in H^{1}(\mathbb{R}^{3})$ such that
$\tilde{v}_{j}^{1}\rightharpoonup\tilde{v}^{1}$ weakly in $H^{1}(\mathbb{R}^{3})$
and strongly in $L_{\text{loc}}^{t}(\mathbb{R}^{3})$ for $2\le t<6$,
and, as a consequence, $v_{j}^{1}\rightarrow\tilde{v}^{1}$ strongly
in $L_{\text{loc}}^{t}(\mathbb{R}^{3})$ for $2\le t<6$, in fact,
fixed $T$, for $\varepsilon_{j}$ sufficiently small $B(0,T)\subset\left\{ |z|<r/\varepsilon_{j}\right\} $.

Given $\varphi\in C_{0}^{\infty}(\mathbb{R}^{3})$, for $\varepsilon_{j}$
small we have that $\text{supp}\varphi\subset B\left(Q_{\varepsilon_{j}}^{1},\frac{r}{2\varepsilon_{j}}\right)$,
thus, since $u_{\varepsilon}$ is a solution of (\ref{eq:kgms}),

\begin{multline}
0=\int_{\text{supp}\varphi}\left[\sum_{il}g_{P^{1}}^{il}(Q_{\varepsilon_{j}}^{1}+\varepsilon_{j}z)\partial_{z_{i}}\tilde{v}_{j}^{1}(z)\partial_{z_{l}}\varphi(z)+(a-\omega^{2})\tilde{v}_{j}^{1}(z)\varphi(z)\right]|g_{P^{1}}((Q_{\varepsilon_{j}}^{1}+\varepsilon_{j}z))|^{1/2}dz\\
-\int_{\text{supp}\varphi}(\tilde{v}_{j}^{1}(z))^{p-1}\varphi(z)|g_{P^{1}}((Q_{\varepsilon_{j}}^{1}+\varepsilon_{j}z))|^{1/2}dz\\
+q\omega^{2}\int_{\text{supp}\varphi}\left(2-q\tilde{\psi}_{j}(z)\right)\tilde{\psi}_{j}(z)\tilde{v}_{j}^{1}(z)\varphi(z)|g_{P^{1}}((Q_{\varepsilon_{j}}^{1}+\varepsilon_{j}z))|^{1/2}dz\label{eq:Q1}
\end{multline}
where $\psi\left(u_{\varepsilon_{j}}\right)(x):=\psi_{j}(x)=\psi_{j}(\exp_{P^{1}}(Q_{\varepsilon_{j}}^{1}+\varepsilon_{j}z)):=\tilde{\psi}_{j}(z)$.
Arguing as in the proof of Proposition \ref{prop:baricentro}, we
can prove that $\tilde{\psi}_{j}\rightarrow0$ in $L_{\text{loc}}^{t}(\mathbb{R}^{3})$
and that $\chi_{r}\tilde{\psi}_{j}\rightharpoonup0$ in $H^{1}(\mathbb{R}^{3})$.
By this, and by (\ref{eq:Q1}) we argue that, for all $\varphi\in C_{0}^{\infty}(\mathbb{R}^{3})$
\[
0=\int_{\text{supp}\varphi}\left[\nabla\tilde{v}^{1}(z)\nabla\varphi(z)+(a-\omega^{2})\tilde{v}^{1}(z)\varphi(z)-(\tilde{v}^{1}(z))^{p-1}\varphi(z)\right]dz,
\]
that is $\tilde{v}^{1}$ weakly solves 
\[
-\Delta\tilde{v}^{1}+(a-\omega^{2})\tilde{v}^{1}=(\tilde{v}^{1})^{p-1}\text{ on }\mathbb{R}^{3}.
\]
We can prove that $v_{j}^{1}\rightarrow\tilde{v}^{1}$ in $C_{\text{loc}}^{2}(\mathbb{R}^{3})$.
We know that, for $z\in B(0,r/2)$,

\begin{multline*}
\sum_{il}\partial_{z_{l}}\left(g_{P^{1}}^{il}(Q_{\varepsilon_{j}}^{1}+\varepsilon_{j}z)|g_{P^{1}}((Q_{\varepsilon_{j}}^{1}+\varepsilon_{j}z))|^{1/2}\partial_{z_{i}}v_{j}^{1}(z)\right)\\
=|g_{P^{1}}((Q_{\varepsilon_{j}}^{1}+\varepsilon_{j}z))|^{1/2}\left\{ -(a-\omega^{2})v_{j}^{1}(z)+(v_{j}^{1}(z))^{p-1}-q\omega^{2}\left(2-q\tilde{\psi}_{j}(z)\right)\tilde{\psi}_{j}(z)v_{j}^{1}(z)\right\} :=f(z);
\end{multline*}
it is easy to see that $f\in L^{\frac{6}{p-1}}$, thus we have
\[
\|v_{j}^{1}\|_{H^{2,\frac{6}{p-1}}(B(0,r/2))}\le\|v_{j}^{1}\|_{L^{\frac{6}{p-1}}(B(0,r/2))}+\|f\|_{L^{\frac{6}{p-1}}(B(0,r/2))}\le C
\]
and we can use a bootstrap argument to prove that indeed $f\in L^{s}$,
and $v_{j}^{1}\in H^{2,s}(B(0,r/2))$ for $s$ sufficiently large,
and, by Sobolev embedding, that $v_{j}^{1}\in C^{0,\theta}(B(0,r/2))$
for some $0<\theta<1$. The same argument may be used to prove that
$\tilde{\psi}_{j}$ is continuous. Then, by the Schauder estimates
(\cite{GT}, page 93), that is 

\[
\|v_{j}^{1}\|_{C^{2,\alpha}(B(0,\tilde{r}))}\le\|v_{j}^{1}\|_{C^{0,\theta}(B(0,r/2))}+\|f\|_{C^{0,\theta}(B(0,r/2))}\le C,
\]
we conclude that $v_{j}^{1}\rightarrow\tilde{v}^{1}$ in $C_{\text{loc}}^{2}(B(0,T))$
for all $T>0$ and, since $\tilde{v}^{1}(0)\ge a-\omega^{2}$, $\tilde{v}^{1}\not\equiv0$.
Thus $\tilde{v}^{1}=U$ positive and radially symmetric. We can repeat
the same argument for $P_{\varepsilon_{j}}^{2}$. Now, choose $\bar{R}$
such that
\[
\int_{B(0,\bar{R})}|\nabla U|^{2}+(a-\omega^{2})U^{2}>\frac{2p}{p-2}\cdot\frac{m_{\infty}+\delta}{2}.
\]
For $\varepsilon_{j}$ sufficiently small, we have that $\varepsilon_{j}\bar{R}\le\frac{d_{g}(P^{1},P^{2})}{2}$,
thus
\begin{eqnarray}
I_{\varepsilon_{j}}(u_{\varepsilon_{j}}) & \ge & \left(\frac{1}{2}-\frac{1}{p}\right)\|u_{\varepsilon_{j}}\|_{\varepsilon_{j}}^{2}\nonumber \\
 & \ge & \left(\frac{1}{2}-\frac{1}{p}\right)\frac{1}{\varepsilon_{j}^{3}}\int_{B_{g}(P^{1},\varepsilon_{j}\bar{R})\cup B_{g}(P^{2},\varepsilon_{j}\bar{R})}\varepsilon^{2}|\nabla_{g}u_{\varepsilon_{j}}|^{2}+(a-\omega^{2})u_{\varepsilon_{j}}^{2}\nonumber \\
 & \rightarrow & 2\left(\frac{1}{2}-\frac{1}{p}\right)\int_{B(0,\bar{R})}|\nabla U|^{2}+(a-\omega^{2})U^{2}>m_{\infty}+\delta\label{eq:2minfty}
\end{eqnarray}
which leads us to a contradiction.\end{proof}
\begin{lem}
\label{lem:unimax}Let $u_{\varepsilon}$ be a solution of (\ref{eq:kgms})
such that $I_{\varepsilon}(u_{\varepsilon})\le m_{\infty}+\delta<2m_{\infty}$.
Then, when $\varepsilon$ is sufficiently small, $u_{\varepsilon}$
has a unique maximum point.\end{lem}
\begin{proof}
Suppose that there exists a sequence $\varepsilon_{j}\rightarrow0$
such that $u_{\varepsilon_{j}}$ has at least two maximum points $P_{\varepsilon_{j}}^{1}$
and $P_{\varepsilon_{j}}^{2}$. By Lemma \ref{lem:2max} we know that
$d_{g}(P_{\varepsilon_{j}}^{1},P_{\varepsilon_{j}}^{2})\rightarrow0$. 

\textbf{Step 1.} It holds 
\begin{equation}
\lim_{j\rightarrow\infty}\frac{1}{\varepsilon_{j}}d_{g}(P_{\varepsilon_{j}}^{1},P_{\varepsilon_{j}}^{2})=+\infty\label{eq:limmax}
\end{equation}
Suppose, by contradiction, that $d_{g}(P_{\varepsilon_{j}}^{1},P_{\varepsilon_{j}}^{2})\le c\varepsilon_{j}$
for some $c>0$. Consider
\[
w_{\varepsilon_{j}}=u_{\varepsilon_{j}}(\exp_{P_{\varepsilon_{j}}^{1}}(\varepsilon_{j}z))\text{ with }|z|\le c.
\]
For any $j$, $w_{\varepsilon_{j}}$ has two maximum points in $B(0,c)$.
Moreover, we can argue, as in Lemma \ref{lem:2max} that $w_{\varepsilon_{j}}\rightarrow U$
in $C_{\text{loc}}^{2}(\mathbb{R}^{3})$ and this is a contradiction.

\textbf{Step 2.} We show that (\ref{eq:limmax}) leads to a contradiction. 

In light of (\ref{eq:limmax}) we have that, fixed $\rho>0$, then
$B_{g}(P_{\varepsilon_{j}}^{1},\rho\varepsilon_{j})\cap B_{g}(P_{\varepsilon_{j}}^{2},\rho\varepsilon_{j})=\emptyset$
for $j$ large. Then we proceed as in Lemma \ref{lem:2max} and we
get the contradiction.\end{proof}
\begin{lem}
\label{lem:stime}Write $u_{\varepsilon}=W_{P_{\varepsilon},\varepsilon}+\Psi_{\varepsilon}$
where $W_{P_{\varepsilon},\varepsilon}$ is defined in (\ref{eq:defweps}).
It holds that $\|\Psi_{\varepsilon}\|_{L^{\infty}(M)}\rightarrow0$
and for any $\rho>0$ $\|u_{\varepsilon}-W_{P_{\varepsilon},\varepsilon}\|_{C^{2}(B_{g}(P_{\varepsilon},\varepsilon\rho))}\rightarrow0$
as $\varepsilon\rightarrow0$.\end{lem}
\begin{proof}
By the $C^{2}$ convergence proved in Lemma \ref{lem:2max} we have
that, given $\rho>0$, 
\[
\|u_{\varepsilon}-W_{P_{\varepsilon},\varepsilon}\|_{C^{2}(B_{g}(P_{\varepsilon},\varepsilon\rho))}=\|u_{\varepsilon}(\exp_{P_{\varepsilon}}(\varepsilon z))-U(z)\|_{C^{2}(B(0,\rho))}\rightarrow0
\]
 as $\varepsilon\rightarrow0$. Moreover, since $u_{\varepsilon}$
has a unique maximum point by Lemma \ref{lem:unimax}, we have that,
for any $\rho>0$, 
\[
\max_{x\in M\smallsetminus B_{g}(P_{\varepsilon},\varepsilon\rho)}u_{\varepsilon}(x)=\max_{x\in\partial B_{g}(P_{\varepsilon},\varepsilon\rho)}u_{\varepsilon}(x)=\max_{|z|=\rho}U(z)+\sigma_{}(\varepsilon)\le ce^{-\alpha\rho}+\sigma_{1}(\varepsilon)
\]
 for some constant $c,\alpha>0$ and for some $\sigma_{1}(\varepsilon)\rightarrow0$
for $\varepsilon\rightarrow0$. This proves the claim.\end{proof}


\begin{thebibliography}{References}
\bibitem{AR}A. Ambrosetti, D. Ruiz, \emph{Multiple bound states for
the Schroedinger-Poisson problem}, Commun. Contemp. Math. \textbf{10}
(2008) 391\textendash{}404.

\bibitem{ADP}A. Azzollini, P. D\textquoteright{}Avenia, A. Pomponio,
\emph{On the Schroedinger-Maxwell equations under the effect of a
general nonlinear term}, Ann. Inst. H. Poincar Anal. Non Linaire \textbf{27}
(2010), no. 2, 779\textendash{}791.

\bibitem{AP}A. Azzollini, A. Pomponio, \emph{Ground state solutions
for the nonlinear Schroedinger-Maxwell equations}, J. Math. Anal.
Appl. \textbf{345} (2008) no. 1, 90\textendash{}108.

\bibitem{AP2}A. Azzollini, A. Pomponio, \emph{Ground state solutions
for the nonlinear Klein-Gordon-Maxwell equations}, Topol. Methods
Nonlinear Anal. \textbf{35} (2010), no. 1, 33\textendash{}42.

\bibitem{BBM}V. Benci, C. Bonanno, and A.M Micheletti.\emph{ On the
multiplicity of solutions of a nonlinear elliptic problem on Riemannian
manifolds}, J. Funct. Anal. \textbf{252} (2007), no. 2, 464\textendash{}489.

\bibitem{BJL} J.Bellazzini, L.Jeanjean, T.Luo, \emph{Existence and
instability of standing waves with prescribed norm for a class of
Schroedinger-Poisson equations} in press on Proc. London Math. Soc.
(arXiv http://arxiv.org/abs/1111.4668)

\bibitem{B}V. Benci, \emph{Introduction to Morse theory: A new approach},
in: Topological Nonlinear Analysis, in: Progr. Nonlinear Differential
Equations Appl., vol. 15, Birkhäuser Boston, Boston, MA, 1995, pp.
37\textendash{}177.

\bibitem{BC1}V. Benci, G. Cerami, \emph{The effect of the domain
topology on the number of positive solutions of nonlinear elliptic
problems}, Arch. Ration. Mech. Anal. \textbf{114} (1991) 79\textendash{}93.

\bibitem{BC2}V. Benci, G. Cerami, \emph{Multiple positive solutions
of some elliptic problems via the Morse theory and the domain topology},
Calc. Var. Partial Differential Equations \textbf{2} (1994) 29\textendash{}48.

\bibitem{BF}V.Benci, D.Fortunato, \emph{An eigenvalue problem for
the Schroedinger-Maxwell equations}, Topol. Methods Nonlinear Anal.
\textbf{11} (1998), no. 2, 283\textendash{}293.

\bibitem{BF1}V. Benci, D. Fortunato, \emph{Solitary waves of the
nonlinear Klein-Gordon field equation coupled with the Maxwell equations},
Rev. Math. Phys. \textbf{14} (2002), no. 4, 409\textendash{}420. 

\bibitem{C}D. Cassani, \emph{Existence and non-existence of solitary
waves for the critical Klein-Gordon equation coupled with Maxwell's
equations}, Nonlinear Anal. \textbf{58} (2004) no. 7-8, 733--747.

\bibitem{CB}Y. Choquet-Bruhat, \emph{Solution globale des Equations
de Maxwell-Dirac-Klein-Gordon}, Rend. Circ. Mat. Palermo \textbf{31}
(1982), no. 2, 267\textendash{}288.

\bibitem{DM}T. D\textquoteright{}Aprile, D. Mugnai, \emph{Solitary
waves for nonlinear Klein-Gordon-Maxwell and Schroedinger-Maxwell
equations}, Proc. Roy. Soc. Edinburgh Sect. A \textbf{134} (2004),
no. 5, 893\textendash{}906.

\bibitem{DM2}T. D'Aprile, D. Mugnai, \emph{Non-existence results
for the coupled Klein-Gordon-Maxwell equations}, Adv. Nonlinear Stud.
\textbf{4} (2004), no. 3, 307\textendash{}322.

\bibitem{DW1}T. D\textquoteright{}Aprile, J. Wei, \emph{Layered solutions
for a semilinear elliptic system in a ball}, J. Differential Equations
\textbf{226} (2006) , no. 1, 269\textendash{}294. 

\bibitem{DW2}T. D\textquoteright{}Aprile, J. Wei, \emph{Clustered
solutions around harmonic centers to a coupled elliptic system}, Ann.
Inst. H. Poincaré Anal. Non Linéaire \textbf{24} (2007), no. 4, 605\textendash{}628.

\bibitem{DP}P. D'Avenia, L. Pisani, \emph{Nonlinear Klein-Gordon
equations coupled with Born-Infeld type equations}, Electron. J. Differential
Equations \textbf{26} (2002), No. 26, 13 pp.

\bibitem{DPS1}P. D'Avenia, L. Pisani, G. Siciliano, \emph{Klein-Gordon-Maxwell
system in a bounded domain}, Discrete Contin. Dyn. Syst. \textbf{26}
(2010), no. 1, 135\textendash{}149.

\bibitem{DPS2}P. D'Avenia, L. Pisani, G. Siciliano, \emph{Dirichlet
and Neumann problems for Klein-Gordon-Maxwell systems}, Nonlinear
Anal. \textbf{71} (2009), no. 12, e1985\textendash{}e1995. 

\bibitem{D}E. Deumens, \emph{The Klein-Gordon-Maxwell nonlinear system
of equations}, Phys. D. \textbf{18} (1986), no. 1-3, 371\textendash{}373. 

\bibitem{DH}O.Druet, E.Hebey, \emph{Existence and a priori bounds
for electrostatic Klein-Gordon-Maxwell systems in fully inhomogeneous
spaces}, Commun. Contemp. Math. \textbf{12} (2010), no. 5, 831\textendash{}869.

\bibitem{GM}M. Ghimenti, A.M. Micheletti, \emph{Low energy solutions
for the semiclassical limit of Schroedinger Maxwell systems, }submitted
paper.

\bibitem{GMP} M. Ghimenti, A.M. Micheletti, A. Pistoia, \emph{The
role of the scalar curvature in some singularly perturbed coupled
elliptic systems on Riemannian manifolds, }submitted paper.

\bibitem{GT}D. Gilbarg, N.S. Trudinger, \emph{Elliptic partial differential
equations of second order. Second edition}. Springer-Verlag, Berlin,
1983. 

\bibitem{IV}I. Ianni, G. Vaira, \emph{On concentration of positive
bound states for the Schroedinger-Poisson problem with potentials},
Adv. Nonlinear Stud.\textbf{ 8} (2008), no. 3, 573\textendash{}595. 

\bibitem{K}Kikuchi, \emph{On the existence of solutions for a elliptic
system related to the Maxwell-Schroedinger equations}, Non- linear
Anal. \textbf{67} (2007) 1445\textendash{}1456.

\bibitem{KM}S. Klainerman, M. Machedon, \emph{On the Maxwell-Klein-Gordon
equation with finite energy}, Duke Math. J. \textbf{74} (1994) 19\textendash{}44.

\bibitem{MN}N. Masmoudi, K. Nakanishi, \emph{Uniqueness of finite
energy solutions for Maxwell-Dirac and Maxwell-Klein-Gordon equations},
Comm. Math. Phys. \textbf{243} (2003), no. 1, 123\textendash{}136. 

\bibitem{M}N. Masmoudi, K. Nakanishi, \emph{Nonrelativistic limit
from Maxwell-Klein-Gordon and Maxwell-Dirac to Poisson-Schroedinger},
Int. Math. Res. Not. \textbf{2003}, no. 13, 697\textendash{}734. 

\bibitem{Mu}D. Mugnai, \emph{Coupled Klein-Gordon and Born-Infeld-type
equations: Looking for solitary waves}, Proc. R. Soc. Lond. Ser. A
Math. Phys. Eng. Sci. \textbf{460} (2004), no. 2045, 1519\textendash{}1527. 

\bibitem{P}R. S. Palais, \emph{Homotopy theory of infinite dimensional
manifolds}, Topology \textbf{5} (1966), 1\textendash{}16.

\bibitem{PS}L.Pisani, G.Siciliano, \emph{Note on a Schroedinger-Poisson
system in a bounded domain}, Appl. Math. Lett. \textbf{21} (2008),
no. 5, 521\textendash{}528. 

\bibitem{R}D. Ruiz, \emph{Semiclassical states for coupled Schroedinger-Maxwell
equations}: Concentration around a sphere, Math. Models Methods Appl.
Sci. \textbf{15} (2005), no. 1, 141\textendash{}164. 

\bibitem{S}G. Siciliano, \emph{Multiple positive solutions for a
Schroedinger-Poisson-Slater system}, J. Math. Anal. Appl. \textbf{365}
(2010), no. 1, 288\textendash{}299. 

\bibitem{WZ}Z. Wang, H.S. Zhou, \emph{Positive solution for a nonlinear
stationary Schroedinger-Poisson system in $\mathbb{R}^{3}$,} Discrete
Contin. Dyn. Syst. \textbf{18} (2007) 809\textendash{}816.\end{thebibliography}
\end{document}